\documentclass[12pt]{article}

\usepackage{amsmath,amscd,amsfonts,eucal,latexsym,amssymb,mathrsfs}
\usepackage{amsxtra,amsthm,calc}
\usepackage{setspace}
\usepackage{mathabx}
\usepackage{graphicx,xcolor}
\usepackage[labelfont=bf]{caption}

\newlength{\dinwidth}
\newlength{\dinmargin}
\newtheorem{Definition}{Definition}[section]
\newtheorem{Theorem}[Definition]{Theorem}
\newtheorem{Proposition}[Definition]{Proposition}
\newtheorem{Lemma}[Definition]{Lemma}

\newcommand*\dif{\mathop{}\!\mathrm{d}}
\newcommand*\Dif[1]{\mathop{}\!\mathrm{d^#1}}

\renewcommand{\theequation}{\thesection.\arabic{equation}}
\allowdisplaybreaks

\usepackage{tikz,pgfplots}
\pgfplotsset{compat=newest}
\usepgfplotslibrary{fillbetween}
\usetikzlibrary{decorations.markings,arrows,arrows.meta,intersections,patterns,patterns.meta,calc,hobby}

\usepackage{authblk}

\title{\Large \bf An approximate local modular quantum energy inequality in general quantum field theory}
\author[1]{Albert Much}
\author[1,2]{Albert Georg Passegger}
\author[1]{Rainer Verch}
\affil[1]{Institute for Theoretical Physics, University of Leipzig, Germany}
\affil[2]{Max Planck Institute for Mathematics in the Sciences, Leipzig, Germany}
\date{}

\begin{document}

\large
\maketitle

\begin{abstract}
	\noindent For every local quantum field theory on a static, globally hyperbolic spacetime of arbitrary dimension, assuming the Reeh-Schlieder property, local preparability of states, and the existence of an energy density as operator-valued distribution, we prove an approximate quantum energy inequality for a dense set of vector states. The quantum field theory is given by a net of von Neumann algebras of observables, and the energy density is assumed to fulfill polynomial energy bounds and to locally generate the time translations. While being approximate in the sense that it is controlled by a small parameter that depends on the respective state vector, the derived lower bound on the expectation value of the spacetime averaged energy density has a universal structure. In particular, the bound is directly related to the Tomita-Takesaki modular operators associated to the local von Neumann algebras. This reveals general, model-independent features of quantum energy inequalities for a large class of quantum field theories on static spacetimes.
\end{abstract}

\section{Introduction}
\label{sec:intro}

In classical (macroscopic) physics, and in particular in general relativity, the energy density of matter is typically non-negative at any point
in spacetime. At a more technical level, this is expressed through 
describing matter by means of a stress-energy tensor field on spacetime
(more commonly, by a stress-energy tensor, for short) on which various 
types of energy conditions are imposed. Such energy conditions are important
in order to ensure that gravity, as formulated in Einstein's equations,
acts always as an attractive force, i.e.\ leads to geodesic focussing. In turn, this is 
crucial for the validity of singularity theorems in general relativity, and also for the ability to rule out certain causal pathologies in solutions
to Einstein's equations of gravity. We will not attempt any review 
or representative selection of  the rich 
amount of literature on the topic and instead just refer to the 
references \cite{Cur,KonSan,Lob} for an overview and review on the subject. 

On the other hand, it is known that in quantum field theory the energy 
density at any given spacetime point is unbounded below as a functional of 
physical quantum states, even if the total energy, i.e.\ the energy density integrated over a Cauchy surface, is non-negative for any quantum field state. This property of the energy density in quantum field theory has
been established both in quantum field models as well as on general grounds.
However, it has been found that suitable spacetime averages of the energy 
density in quantum field theory or, more technically, integration of 
the energy density against smooth, non-negative test functions, leads to quantities 
(expectation value functionals) which are bounded below, in such a manner that
the bound does not depend on the states (with respect to which the expectation values are formed). 
We will next mention some investigations and results to this effect, however again without attempting a fair review of the literature; considerably more discussion to this 
end is given in \cite{Few-QEIRev-FTPh,KonSan} and references cited there.\medskip

The object of study is a quantum field theory with a stress-energy observable, which
in general is a tensorial operator-valued distribution. Its expectation value for 
a suitable class of states is, somewhat symbolically, denoted by $\langle T_{\mu \nu}(x)\rangle_\sigma$ where $\sigma$ labels a state, and $x$ is a spacetime point. 
This sloppy notation for a distribution is common and justified since, under general
conditions, the expectation value is actually a smooth function on spacetime for 
sufficiently regular states $\sigma$. For many types of quantum fields that 
are subject to a linear hyperbolic equation of motion, or a Dirac-type equation,
there are good candidates for the expectation value of the stress-energy tensor, in flat as well in general curved spacetimes, for a class of states $\sigma$ that 
fulfill the microlocal spectrum condition, which can be seen as a generalization 
of the Hadamard condition on quantum field states \cite{Rad,BruFreKoe,SahVer}. We won't review the matter of definition of the quantum stress-energy tensor on curved spacetime here and refer to \cite{Wald-QFTCST,Mor} for further discussion. 
For linear quantum fields on curved spacetimes, it has been shown 
that a {\it quantum weak energy inequality} (QWEI)
holds, i.e.\ if $\gamma$ is a smooth timelike curve, and if 
$\langle \varrho(t) \rangle_\sigma = \langle T_{\mu \nu}(\gamma(t))\rangle_\sigma\dot{\gamma}^\mu(t)
\dot{\gamma}^\nu(t)$ denotes the expected energy density along the curve in state $\sigma$, then an estimate of the type
\begin{align}
 \inf_\sigma \int f^2(t)\, \langle \varrho(t) \rangle_\sigma \dif t \ge -C(\gamma,f)
\end{align}
holds for real-valued, smooth, compactly supported test functions $f$ along $\gamma$
\cite{Few2000,FV-Dirac,FewKon}.
The crucial point is that the right hand side is potentially negative, but finite
for every $f$ and every $\gamma$.
This finding comes with a caveat, though, in that it may depend on the type of 
field equation. For instance, it has been found that the non-minimally coupled 
linear scalar field in general does not admit a bound of this type. In fact, 
for such quantum fields, a weaker statement, called a {\it relative quantum weak energy inequality} (rQWEI) holds, of the type 
\begin{align}
 \int f^2(t)\, \langle \varrho(t) \rangle_\sigma \dif t \ge -Q_\sigma(\gamma,f) \, ,
\end{align}
where the functional $Q_\sigma(\gamma,f) \ge 0$ provides a lower bound in the given form, but not
an upper bound on the integral on the left hand side as $\sigma$ ranges over 
the set of all Hadamard states for fixed $\gamma$ and $f$  
\cite{FewOst}. 
Limiting behaviour of the QWEI has also been investigated, for the case that 
the function $f$ approaches the constant value $1$ all along the curve 
$\gamma$. If, in this case, the lower bound on the resulting expression is 
$0$ for all states contemplated, one speaks of an {\it averaged weak energy condition} (AWEC) if $\gamma$ is a complete timelike
geodesic, and of an {\it averaged null energy condition} (ANEC) if $\gamma$ is a complete
null geodesic. For free quantized fields on Minkowski spacetime, AWEC results 
can be obtained as limiting cases of the QWEI \cite{ForRo}; however it has been
pointed out that the AWEC may fail in the presence of negative static Casimir 
energies as they occur near domain walls, which may indirectly be taken as 
a failure of AWEC in interacting quantum field theories possessing bound states 
which may serve as domain walls \cite{GraOlum}. Results on the ANEC have been obtained
for free fields in Minkowski spacetime \cite{Kli,Wald-Yur,Fol} and also for interacting 
quantum fields in two-dimensional Minkowski spacetime \cite{FewHol,Ver-ANEC}. 
In higher dimensional Minkowski spacetime, lower bounds on the null energy density 
of the expected stress-energy tensor, including ANEC, have been connected to bounds on entropy-like quantities for certain spacetime regions in general quantum field theories \cite{Wall,BFKLW,FLPW,Lon,CeyFau,MoTaWe,CaGriPo}. From the perspective of our present contribution, there is some potential link since methods from the Tomita-Takesaki modular theory of von Neumann algebras \cite{Tak,BraRob} are used in the mentioned works. One of the prominent appearances of the Tomita-Takesaki modular theory in general quantum field theory is via the theorems of Bisognano-Wichmann and Borchers on the geometric 
action of Tomita-Takesaki modular objects for operator algebras of observables 
associated to certain (wedge) regions and the vacuum vector \cite{BiWi,Bor92,
Bor-Revol,Haag}. Tomita-Takesaki modular objects of local observable algebras relative to a vacuum state
also make an appearance in the main result of the present work, however without 
need for their geometric action. 

While there are numerous results on QWEIs and also averaged energy inequalities for
linear quantum field theories, including such on general spacetimes, there is 
apparently little on locally averaged lower bounds on energy expectation values 
for general and interacting quantum field theories. The results in this direction
so far have been relatively sparse \cite{BosFew,BosCadFew,Cad-QEI-Int-Rev}, reflecting the difficulty of obtaining
local energy expressions permitting useful bounds for interacting quantum fields. Moreover, it is unclear 
what types of locally averaged energy inequalities one may hope to expect. Certainly
one wouldn't expect locally spatially averaged energy quantities to hold as they 
fail already for linear quantum fields in physical spacetime dimension
\cite{ForHelRo}. 
One may even be skeptical about rQWEIs in the light of the arguments of \cite{GraOlum}, even though it is not clear how relevant they are at a formalized mathematical level. As pointed out in \cite{FewPfe}, lower bounds on locally spacetime averaged 
energy expectation values have a potentially larger domain of applicability. This 
is basically also the starting point we adopt for the present work. \medskip

A first step, and difficulty to some extent, is to settle for definite assumptions for a general
quantum field theory with a stress-energy tensor, or a bit more specifically, an energy density. To this end, we follow in this work largely the approach of 
\cite{Ver-ANEC} and \cite{FV-Pass}. We consider a quantum field theory in the operator algebraic setting on a static, globally hyperbolic spacetime. Standard assumptions are made, such as commutativity of local observables at causal separation, existence of a ground state for the time translations, the Reeh-Schlieder property for local algebras, and local preparability of states (which is related to the split property and the type III property of local algebras \cite{Wer-LocPreSplit,Yng-typeIII,Sum,Few-split-rev}). These assumptions have been shown to be fulfilled for a large class of linear quantum field theories (see Section \ref{sec:assump} for a discussion). Furthermore, it will be assumed that there is an operator-valued
distribution $\boldsymbol{\varrho}(F)$, $F \in C_0^\infty(M)$, on a suitable common dense domain, affiliated with the local observable algebras, fulling polynomial ``$H$-bounds'', where $H$ denotes the Hamilton operator of the quantum field theory, and generating the time translations locally. The latter means that if $G$ is a test function of a certain type, and $A$ is a 
local observable with a suitable localization relative to $G$, then $[\boldsymbol{\varrho}(G),A] = [H,A]$. (A fully rigorous formulation of the property will be given in Section \ref{sec:setting}. As usual, $[X,Y] = XY -YX$ denotes the algebraic commutator.) Using these assumptions, we reach at the following main result. 
Let $\Omega$ denote the unit vector inducing the ground state, and let $O^\sharp$ be
a spacetime region which is suitably larger than ${\rm supp}(G)$. Then for every 
unit vector $\psi$ in the common dense domain of all energy density operators,
it holds that for every $\epsilon > 0$ there is some $\lambda_0 > 0$ so that 
\begin{align} \label{eq:qei.}
 (\psi,\boldsymbol{\varrho}(G)\psi) \ge - \epsilon - \sqrt{2\pi}\,\| {\rm e}^{-(\lambda K_\sharp)^2/2}
 \Delta^{-1/2}_\sharp \boldsymbol{\varrho}(G)\Omega\|
\end{align}
for all $0 < \lambda < \lambda_0$, where 
$\Delta_\sharp = {\rm e}^{K_\sharp}$ denotes the Tomita-Takesaki modular operator 
associated to the local observable algebra ${\sf A}(O^\sharp)$ and $\Omega$; note that $\Delta_\sharp^{-1/2} = {\rm e}^{-K_\sharp/2}$.
The round brackets are used for the scalar product of the ambient Hilbert space (the ground state Hilbert space). The number $\lambda_0 > 0$ depends on the state vector $\psi$ and on $\epsilon$.
The dependence on $\lambda_0$ means that the bound \eqref{eq:qei.} is not a state-independent quantum inequality as in the QWEIs for linear quantum fields. It 
has more the character of a rQWEI, where the state dependence enters through the 
dependence on $\lambda_0$. It should be noted that $\boldsymbol{\varrho}(G)\Omega$
cannot be expected to lie in the domain of $\Delta^{-1/2}_\sharp$; a result to that effect will be provided in Appendix \ref{appendix-B}. Thus, the expression on the right hand side of \eqref{eq:qei.} will diverge to $-\infty$ for $\lambda \to 0$. What is of interest, however, is the ``universal'' appearance of the expression, which not only is independent of the state vector $\psi$, but actually of any details of the energy density, and the way it is tied to the Tomita-Takesaki modular operators associated to local von Neumann algebras of the quantum field theory and the ground state vector. \medskip

This work is organized as follows. In Section \ref{sec:setting} we introduce our setup of a quantum field theory on a static, globally hyperbolic spacetime in operator algebraic language. In particular, we define the concept of an energy density, subject to certain (physically motivated) conditions. Section \ref{sec:assump} comprises a discussion of the imposed assumptions with regard to their generality and validity in quantum field theory models, as well as a standard approximation result that will be important in the proof of the main theorem in Section \ref{sec:aqei}. In Theorem \ref{Thm:First} we first present a preliminary, simple quantum energy inequality for the locally averaged energy density with respect to non-dense bounded sets of unit vectors. The main result, Theorem \ref{Thm:3rd}, constitutes a refined version of this inequality with a lower bound for a dense set of state vectors, as sketched above. This result is followed by a concluding discussion in Section \ref{sec:discussion}. Two technical appendices appear after the main body of the article.

\section{Setting}
\label{sec:setting}

We consider a quantum field theory on a $(1+d)$-dimensional static,  
globally hyperbolic spacetime $(M,g)$ with manifold $M = \mathbb{R} \times \Sigma$,
where $\Sigma$ is a $d$-dimensional manifold. Points in $M$ will generically be denoted by $(t,p)$ with $t \in \mathbb{R}$ and $p \in \Sigma$. The metric $g$ on the spacetime is given by
\begin{align}
	\label{eq:metric}
	g = \alpha\dif t^2 \oplus (-h) \, ,
\end{align}
where $h$ is a Riemannian metric on $\Sigma$ and $\alpha$ is a smooth, strictly positive function on $\Sigma$. The level sets of the time function $(t,p) \mapsto t$, i.e.\ the sets $\Sigma_t = \{t\} \times \Sigma$ $(t \in \mathbb{R})$, are assumed to be Cauchy surfaces. There is a Killing flow $\{\tau_s\}_{s \in \mathbb{R}}$ on $M$ associated to the global timelike Killing vector field $\partial_t$, consisting of the time shift isometries $\tau_s(t,p) = (t + s,p)$. We also recall the following notation. For any subset $T \subset M$, 
$J^\pm(T)$ are the causal future ($+$) and causal past ($-$) sets of $T$, respectively. The {\it domain of dependence} (also called Cauchy development) $D(T)$ of $T$ consists of those points for which all future- or past-inextendible causal curves through them intersect $T$. Finally, we call $T$ {\it causally convex} if it agrees with its causal hull $J^+(T)\cap J^-(T)$, i.e.\ if every causal curve with endpoints in $T$ is contained in $T$. We refer to \cite{ONeill,Wald-GR} for further background on Lorentzian geometry and discussion of these concepts. 
\medskip

The quantum field theory on a static spacetime $(M,g)$ is described in the 
model-independent, operator algebraic framework \cite{BauWo,Haag,FV-Pass}. It is assumed that there
is a family of von Neumann algebras $\{{\sf A}(O)\}_{O \subset M}$ on a separable Hilbert space $\mathcal{H}$, indexed by the open 
subsets $O$ of $M$. We list below the properties we will impose on our quantum field theory, while the next section will contain a discussion as to how general these assumptions are, and the extent to which they are proven, or can be expected to hold, in specific quantum field theory models. We also present some further consequences of the assumptions. 
\begin{itemize}
 \item[(a)] {\it Isotony}: $O_1 \subset O_2 \Rightarrow {\sf A}(O_1) \subset {\sf A}(O_2)$
 \item[(b)] {\it Locality}: $O_1 \subset O^\perp \Rightarrow {\sf A}(O_1) \subset {\sf A}(O)'$
 \\
 Here, $O^\perp = M \setminus (J^+(O) \cup J^-(O))$ denotes the {\it causal complement} of $O$, i.e.\ the set of all spacetime points which cannot be connected to $O$ by any causal curve. Furthermore, for any ${\sf N} \subset \mathcal{B}(\mathcal{H})$, the commutant of ${\sf N}$ is denoted by ${\sf N}' = \{X \in \mathcal{B}(\mathcal{H}): XN = NX \ \text{for all} \ N \in {\sf N} \,\}$.
 
 \item[(c)] {\it Covariance}: There is a strongly continuous one-parameter unitary group $\{U_t\}_{t \in \mathbb{R}}$ on $\mathcal{H}$ so that $U_t {\sf A}(O)U_t^{-1} = {\sf A}(\tau_t(O))$ holds for all $t \in \mathbb{R}$ and all open subsets $O$ of $M$. 
 \item[(d)] {\it Existence of a ground state}: There is a unit vector $\Omega \in \mathcal{H}$ so that $U_t\Omega = \Omega$ for all $t \in \mathbb{R}$. Moreover, denoting by $H$ the ``Hamiltonian'',  that is, the selfadjoint generator of $\{U_t\}_{t \in \mathbb{R}}$ given by $U_t = {\rm e}^{itH}$, it holds that ${\rm spec}(H) \subset [0,\infty)$ (the spectrum of $H$ contains no negative values).
 \item[(e)] {\it Reeh-Schlieder property}: If $O$ is any non-empty open subset of $M$,
 then $\Omega$ is cyclic for ${\sf A}(O)$, meaning that the set 
 ${\sf A}(O)\Omega = \{ A\Omega: A \in {\sf A}(O)\}$ is dense in $\mathcal{H}$. This implies that $\Omega$ is also separating for ${\sf A}(O)$ whenever there is an non-empty open subset $O_1$ of $M$ with $O_1 \subset \overline{O^\perp}$; we recall that $\Omega$ is called separating for ${\sf A}(O)$ if 
 for every $A \in {\sf A}(O)$ the equation $A\Omega = 0$ implies $A = 0$. 
 \item[(f)] {\it Local preparability of states}: If $O$ and $O_1$ are causally convex open subsets of $M$ with $\overline{O_1} \subset O$, then for every unit vector $\psi \in \mathcal{H}$ there is some $Y \in {\sf A}(O)$ with $\| Y \| = 1$ such that 
\begin{align}
 (\psi, A \psi) = (Y\Omega,A Y\Omega) \quad (A \in {\sf A}(O_1))\,.
\end{align}
\end{itemize}
It will be helpful to introduce further notation. For $f \in \mathscr{S}(\mathbb{R})$
(the Schwartz functions) we write
\begin{align} \label{eq:convol}
u_f(A) = \int_{-\infty}^\infty f(t) U_t A U_t^{-1} \dif t
\end{align}
whenever $A \in \mathcal{B}(\mathcal{H})$. We then write ${\sf A}_\infty(O)$ for the 
$*$-subalgebra of ${\sf A}(O)$ formed by finite polynomials of 
elements of the form $u_f(B)$ (required to be in ${\sf A}(O)$) for $f\in\mathscr{S}(\mathbb{R})$ and $B \in {\sf A}(O)$. Note that 
\begin{align}
\left. \frac{\dif}{\dif t} U_t u_f(B) U_t^{-1}\right|_{t = 0} = i[H,u_f(B)] = u_{\dot{f}}(B) 
\end{align}
where $\dot{f} = \dif f/\dif t$. Therefore, if $A \in {\sf A}_\infty(O)$, then 
$H^n A\Omega$ lies in $\mathcal{H}$ for every $n \in \mathbb{N}$. Put differently, $A\Omega$
is in the $C^\infty$-domain of $H$ whenever $A \in {\sf A}_\infty(O)$. 
\medskip

The further important assumption we make is that there is an energy density given
by a quantum field which generates the derivation of the Hamiltonian $H$ locally, satisfies
a polynomial $H$-bound and is affiliated with the local von Neumann algebras ${\sf A}(O)$.
In more detail:

\begin{itemize}
\item[(A)] We assume that for any $F \in C_0^\infty(M)$, there is a linear operator 
$\boldsymbol{\varrho}(F)$, depending linearly on $F$, defined on
a common dense domain $\mathcal{D} \subset \mathcal{H}$, and that $\boldsymbol{\varrho}(F)$
is essentially selfadjoint on $\mathcal{D}$ if $F$ is real-valued. 
The domain $\mathcal{D}$ is assumed to be invariant under the action of the $\boldsymbol{\varrho}(F)$ and the $U_t$, and to contain $\Omega$. 
Moreover, covariance
will be assumed: 
\begin{align}
	\label{eq:covariance-rho}
	U_t \boldsymbol{\varrho}(F) U_t^* = \boldsymbol{\varrho}(F \circ \tau_{-t}) \quad (F \in C_0^\infty(M),\ t \in \mathbb{R})\,.
\end{align}
\item[(B)] Operator-valued distribution: For every $m \in \mathbb{N}$ and $\psi \in \mathcal{D}$, the map 
\begin{align}
	F_1 \otimes \cdots \otimes F_m \mapsto (\psi,\boldsymbol{\varrho}(F_1) \cdots \boldsymbol{\varrho}(F_m) \psi) \quad (F_j \in C_0^\infty(M))
\end{align}
extends linearly to a distribution on $M^m$.
\item[(C)] Furthermore, it will be assumed that there is an integer $\ell$ so that 
$(1 + H)^{-\ell} \boldsymbol{\varrho}(F) (1 + H)^{-\ell}$ extends, for any $F \in C_0^\infty(M)$, to a bounded operator on $\mathcal{H}$.
This is  referred to by saying that ``the 
$\boldsymbol{\varrho}(F)$ fulfill a polynomial $H$-bound''. It implies that 
$\boldsymbol{\varrho}(F)(1 + H)^{2 \ell}$ extends to a bounded operator for all
$F \in C_0^\infty(M)$ \cite{FreHer}. Therefore, all vectors in ${\rm dom}( (1+ H)^{2\ell})$,
the (graph norm closed) domain of definition of $(1 + H)^{2\ell}$, are in the 
(graph norm closed) domain of definition ${\rm dom}(\boldsymbol{\varrho}(F))$
of the selfadjoint extension of $\boldsymbol{\varrho}(F)$ for real-valued $F$. 
\item[(D)] It is also assumed that the $\boldsymbol{\varrho}(F)$ are affiliated with the 
local von Neumann algebras in the following sense:
For real-valued $F$ such that ${\rm supp}(F) \subset O$, it holds that every bounded function 
$b(\boldsymbol{\varrho}(F))$ of $\boldsymbol
{\varrho}(F)$ in the sense of the spectral calculus (for $b: \mathbb{R} \to \mathbb{R}$ continuous and bounded)
is contained in ${\sf A}(O)$.\footnote{We notationally identify $\boldsymbol{\varrho}(F)$ and its selfadjoint extension as no ambiguity is likely to arise.}
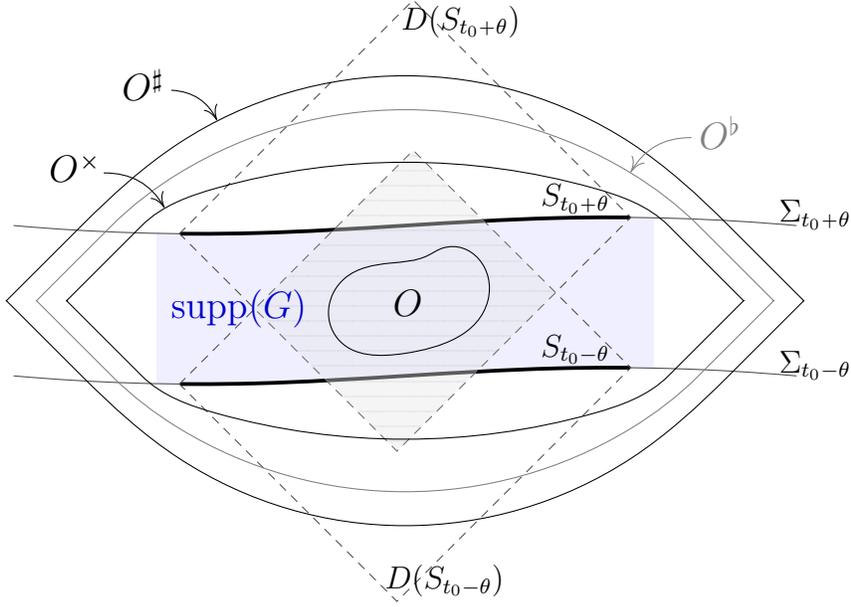
\begin{figure}[!ht]
	\centering
	\begin{tikzpicture}[use Hobby shortcut,tangent/.style={in angle={(180+#1)},Hobby finish,designated Hobby path=next,out angle=#1},extended line/.style={shorten >=-#1,shorten <=-#1},fill between/on layer=main,line join=miter,line cap=rect]
		\useasboundingbox (-5.7,-4.2) rectangle (6.3,4.2);
		%
		%%% Cauchy surfaces def.
		%
		\path[name path=upper] (-5.2,1)..(-2,0.9)..(2,1.1)..(5.2,1); % upper surface
		\path[name path=lower] (-5.2,-1)..(-2,-1.1)..(2,-0.9)..(5.2,-1); % lower surface
		%
		%%% Support of $G$
		%
		\begin{scope}
			\clip (-3.3,-3) rectangle (3.3,3);
			\tikzfillbetween[of=upper and lower]{blue!30,opacity=0.2};
		\end{scope}
		%
		%%% Cauchy surfaces draw
		%
		\draw[line cap=round,black!70] (-5.2,1)..(-2,0.9)..(2,1.1)..(5.2,1);
		\draw[line cap=round,black!70] (-5.2,-1)..(-2,-1.1)..(2,-0.9)..(5.2,-1);
		\path[name path=r1] (3,-3) -- (3,3); % help line vertical right (boundary S)
		\path[name path=l1] (-3,-3) -- (-3,3); % help line vertical left (boundary S)
		%
		% thicker lines for $S$
		\begin{scope}
			\clip (-2.983,-3) rectangle (2.983,3);
			\draw[line width=0.5mm] (-5.2,1)..(-2,0.9)..(2,1.1)..(5.2,1);
			\draw[line width=0.5mm] (-5.2,-1)..(-2,-1.1)..(2,-0.9)..(5.2,-1);
		\end{scope}
		%
		%%% Define intersection points $S\subset\Sigma$
		%
		\node[coordinate,name intersections={of=upper and l1}] (ul) at (intersection-1) {};
		\node[coordinate,name intersections={of=upper and r1}] (ur) at (intersection-1) {};
		\node[coordinate,name intersections={of=lower and l1}] (dl) at (intersection-1) {};
		\node[coordinate,name intersections={of=lower and r1}] (dr) at (intersection-1) {};
		%
		%%% Cauchy developments
		%
		% Upper
		\path[name path=ulpath1] (ul) -- ([yshift=5cm, xshift=5cm]ul);
		\path[name path=urpath1] (ur) -- ([yshift=5cm, xshift=-5cm]ur);
		\path[name path=ulpath2] (ul) -- ([yshift=-5cm, xshift=5cm]ul);
		\path[name path=urpath2] (ur) -- ([yshift=-5cm, xshift=-5cm]ur);
		\node[coordinate,name intersections={of=ulpath1 and urpath1}] (uc1) at (intersection-1) {};
		\node[coordinate,name intersections={of=ulpath2 and urpath2}] (uc2) at (intersection-1) {};
		\draw[name path=ucauchy,dashed,black!70] (ul) -- (uc1) -- (ur) -- (uc2) -- cycle;
		%
		% Lower
		\path[name path=dlpath1] (dl) -- ([yshift=5cm, xshift=5cm]dl);
		\path[name path=drpath1] (dr) -- ([yshift=5cm, xshift=-5cm]dr);
		\path[name path=dlpath2] (dl) -- ([yshift=-5cm, xshift=5cm]dl);
		\path[name path=drpath2] (dr) -- ([yshift=-5cm, xshift=-5cm]dr);
		\node[coordinate,name intersections={of=dlpath1 and drpath1}] (dc1) at (intersection-1) {};
		\node[coordinate,name intersections={of=dlpath2 and drpath2}] (dc2) at (intersection-1) {};
		\draw[name path=dcauchy,dashed,black!70] (dl) -- (dc1) -- (dr) -- (dc2) -- cycle;
		%
		% Fill intersection of Cauchy developments
		\node[coordinate,name intersections={of=ucauchy and dcauchy}] (rcauchy) at (intersection-1) {};
		\node[coordinate,name intersections={of=ucauchy and dcauchy}] (lcauchy) at (intersection-2) {};
		\fill[black!10,opacity=0.4,postaction={pattern={Lines[angle=45, line width=0.2pt, distance=2mm]},pattern color=black!50}] (lcauchy) -- (dc1) -- (rcauchy) -- (uc2) -- cycle;
		%
		%%% Subsets
		%
		% O
		\coordinate (A) at (1.1,0);
		\coordinate (B) at (0.5,0.7);
		\coordinate (C) at (0.3,0.6);
		\coordinate (D) at (-0.2,0.5);
		\coordinate (E) at (-0.8,0.3);
		\coordinate (F) at (-0.9,-0.5);
		\coordinate (G) at (0,-0.7);
		\draw[black] (A) to [closed, curve through = {(B) (C) (D) (E) (F) (G)}] (A);
		%
		% O^\times
		\coordinate (A1) at (4.5,0);
		\coordinate (B1) at (3.5,1);
		\coordinate (C1) at (2.5,1.5);
		\coordinate (D1) at (-2.5,1.5);
		\coordinate (E1) at (-3.5,1);
		\coordinate (F1) at (-4.5,0);
		\coordinate (G1) at (-3.5,-1);
		\coordinate (H1) at (-2.5,-1.5);
		\coordinate (I1) at (2.5,-1.5);
		\coordinate (J1) at (3.5,-1);
		\draw[black] (A1)..([tangent=135]B1)..(C1)..(D1)..([tangent=225]E1)..(F1) (F1)..([tangent=315]G1)..(H1)..(I1)..([tangent=45]J1)..(A1);
		%
		% O^$\flat$
		\coordinate (A2) at (4.9,0);
		\coordinate (B2) at (3.9,1);
		\coordinate (C2) at (3,1.7);
		\coordinate (D2) at (-3,1.7);
		\coordinate (E2) at (-3.9,1);
		\coordinate (F2) at (-4.9,0);
		\coordinate (G2) at (-3.9,-1);
		\coordinate (H2) at (-3,-1.7);
		\coordinate (I2) at (3,-1.7);
		\coordinate (J2) at (3.9,-1);
		\draw[black!50] (A2)..([tangent=135]B2)..(C2)..(D2)..([tangent=225]E2)..(F2) (F2)..([tangent=315]G2)..(H2)..(I2)..([tangent=45]J2)..(A2);	
		%
		% O^\sharp
		\coordinate (A3) at (5.3,0);
		\coordinate (B3) at (4.3,1);
		\coordinate (C3) at (3,2.1);
		\coordinate (D3) at (-3,2.1);
		\coordinate (E3) at (-4.3,1);
		\coordinate (F3) at (-5.3,0);
		\coordinate (G3) at (-4.3,-1);
		\coordinate (H3) at (-3,-2.1);
		\coordinate (I3) at (3,-2.1);
		\coordinate (J3) at (4.3,-1);
		\draw[black] (A3)..([tangent=135]B3)..(C3)..(D3)..([tangent=225]E3)..(F3) (F3)..([tangent=315]G3)..(H3)..(I3)..([tangent=45]J3)..(A3);	
		%
		%%% Labels
		%
		\node[scale=1.2] at ([xshift=1pt,yshift=-1pt]0,0) {$O$};
		\node[scale=1.2] at ([xshift=-6pt]lcauchy) {$\textcolor{blue!80!black}{\mathrm{supp}(G)}$};
		\node at ([xshift=-21pt,yshift=7pt]ur) {$S_{t_0 + \theta}$};
		\node at ([xshift=-21pt,yshift=7pt]dr) {$S_{t_0 - \theta}$};
		\node at ([xshift=27pt,yshift=4pt]4.5,1) {$\Sigma_{t_0 + \theta}$};
		\node at ([xshift=27pt,yshift=4pt]4.5,-1) {$\Sigma_{t_0 - \theta}$};
		\node at ([xshift=18pt,yshift=-8pt]uc1) {$D(S_{t_0 + \theta})$};
		\node at ([xshift=18pt,yshift=8pt]dc2) {$D(S_{t_0 - \theta})$};
		\path[draw=black,{<[width=2mm,length=2mm]}-] (-3.2,1.24) to[bend right] (-4,1.7) node[xshift=-11pt,yshift=2pt,scale=1.2] {$O^\times$};
		\path[draw=black!50,{<[width=2mm,length=2mm]}-] (3,1.71) to[bend left] (3.8,2.17) node[xshift=11pt,yshift=2pt,scale=1.2] {$\textcolor{black!50}{O^\flat}$};
		\path[draw=black,{<[width=2mm,length=2mm]}-] (-2.5,2.4) to[bend right] (-3.1,2.8) node[xshift=-11pt,yshift=2pt,scale=1.2] {$O^\sharp$};
	\end{tikzpicture}
	\caption{Illustration of the relative localization of ${\rm supp}(G)$ (shaded in blue), the intersection $D(S_{t_0 - \theta}) \cap D(S_{t_0 + \theta})$ (shaded and hatched in grey), and $O$ according to assumption (E). The time coordinate $t$ runs vertically upwards. The causally convex open spacetime regions $O^\times$, $O^\flat$ and $O^\sharp$ appear in Theorem \ref{Thm:3rd}.}
	\label{fig}
\end{figure}
\item[(E)] The property giving $\boldsymbol{\varrho}$ the significance of an energy density is the following (see Figure \ref{fig} for a depiction of the geometric setup). Suppose that $\underline{g} \in C_0^\infty(\Sigma)$ is non-negative, fulfilling
$\underline{g}(p) = 1$ for all $p \in S$ with $S$ an open subset of $\Sigma$, and let 
$g_0 \in C_0^\infty(\mathbb{R})$ be non-negative, with ${\rm supp}(g_0) = [t_0 -\theta, t_0 + \theta]$
for some $t_0 \in \mathbb{R}$ and some $\theta >0$, and $\int_{-\infty}^\infty g_0(t)\dif t = 1$. Then for $G$ defined by $G(t,p) = g_0(t)\underline{g}(p)$ and any $O \subset D(S_{t_0 - \theta}) \cap D(S_{t_0 + \theta}) \cap {\rm supp}(G)$, it is assumed that 
\begin{align} \label{eq:dynamicgenerator}
	[\boldsymbol{\varrho}(G),A]\Omega = [H,A]\Omega = HA \Omega \quad (A \in {\sf A}_\infty(O))\,.
\end{align}
Here we have used the notation $S_t = \{t\} \times S$, and also $H\Omega = 0$ (assumption (d)) in the rightmost equality of \eqref{eq:dynamicgenerator}.
Note furthermore that the condition $O \subset D(S_{t_0 - \theta}) \cap D(S_{t_0 + \theta})$ implicitly
requires a small enough $\theta$ and a large enough $S$.
\end{itemize}

\section{Discussion of the assumptions}
\label{sec:assump}

In the present section, we will provide some more comments on the assumptions made, as well as some additional results which will be used later. \medskip

The assumptions (a)--(d) are quite standard for algebraic quantum field theory. They are generalizations of the Haag-Kastler axioms for an opera\-tor-algebraic setting of a quantum field theory in a vacuum representation on Minkowski spacetime \cite{Haag,HaKa} to a theory in the GNS representation of a ground state on a static, globally hyperbolic spacetime, where they are typical properties in the example case of linear quantum field theories \cite{Dim,Wald-QFTCST,Sanders-KMS} (see also \cite{FV-AQFT} for a more general context).

Assumption (e), the Reeh-Schlieder property, was first shown to be fulfilled for the quantized free scalar field on Minkowski spacetime \cite{ReehSch,StrWigh}. It expresses the existence of correlations in the vacuum state, enabling the approximation of any state with arbitrary accuracy by means of operations on the vacuum that are localized in any open region. This property is also known to hold, among other cases, for ground states of certain linear quantum fields on static, globally hyperbolic spacetimes \cite{Str-RS}, and in a weaker form for locally covariant quantum field theories on globally hyperbolic spacetimes, including Klein-Gordon, Dirac and Proca fields \cite{San-RS,Dap-RS}.

Assumption (f) on the local preparability of states is typically inferred from the split property and the type III property of the local von Neumann algebras; we refer to the reviews \cite{Yng-typeIII,Yng-LocEnt,Sum} and also \cite{BuDoLo,Wer-LocPreSplit}. The split property, which implies statistical independence of local algebras associated to spacelike separated regions, has been established in the GNS representation of quasifree Hadamard states of some linear quantum fields on globally hyperbolic spacetimes \cite{Ver-split,Ver-continuity,DAnHol}, and under certain conditions also for general locally covariant quantum field theories \cite{Few-split} (see also \cite{Few-split-rev}). The type III property (in the classification of von Neumann factors, see \cite[Sec. V.2.4]{Haag} for a brief outline) is a common property of local von Neumann algebras in quantum field theory that has been shown to be fulfilled, e.g., in the GNS representation of quasifree Hadamard states of the Klein-Gordon and Dirac field on curved spacetimes \cite{Ver-continuity,DAnHol}; further general results, sufficient conditions and references can be found in \cite{Dri,Fre,BuVe,Yng-typeIII}.\medskip

Concerning the assumptions (A)--(E) on the energy density, they are largely the assumptions one would make for a local quantum field in an operator-algebraic context, where assumption (D) expresses the locality condition in a strong form. The ``$H$-bound'' of assumption (C) has been widely discussed in the said context, and it holds under general conditions for quantum field theories on Minkowski spacetime; see Chs.\ 12 to 14 in \cite{BauWo} as well as \cite{FreHer} and references cited there for further discussion. 

One of the consequences is that, for the case of (say, $(1 + 3)$-dimensional) Minkowski spacetime, the domain $\mathcal{D}$ can be chosen such that for any $\psi,\psi' \in \mathcal{D}$ there is a smooth function $x \mapsto (\psi,\boldsymbol{\varrho}[x]\psi')$ of spacetime points $x$ so that $(\psi,\boldsymbol{\varrho}(F)\psi') = 
\int_{\mathbb{R}^4} F(x)\,(\psi,\boldsymbol{\varrho}[x]\psi')\Dif4 x$ holds for all smooth, compactly supported test functions $F$ on Minkowski spacetime. Note that 
$(\psi,\boldsymbol{\varrho}[x]\psi')$ is a slightly improper notation for a quadratic form on $\mathcal{D}$ defined for every $x$. In a similar manner, one 
can assume that every coordinate component of a stress-energy tensor observable of the quantum field theory is given by an $H$-bounded quantum field $F \mapsto \boldsymbol{T}_{\mu\nu}(F)$, and that for any $\psi,\psi' \in \mathcal{D}$ the map $x \mapsto (\psi,\boldsymbol{T}_{\mu \nu}[x]\psi')$ is a smooth function on spacetime satisfying $(\psi,\boldsymbol{T}_{\mu \nu}(F)\psi') = \int_{\mathbb{R}^4} F(x) (\psi,\boldsymbol{T}_{\mu \nu}[x]\psi')\Dif4 x$. Actually,
the expectation values of the (renormalized) stress-energy tensor in Hadamard states
of linear quantum fields on generic spacetimes are given by smooth tensor fields \cite{Mor,Wald-QFTCST}.

For the $(1+d)$-dimensional static, globally hyperbolic spacetime $M = \mathbb{R} \times \Sigma$ that we consider, we envisage the energy density as arising from a stress-energy tensor in the sense that 
\begin{align} \label{eq:stress_00}
 (\psi,\boldsymbol{\varrho}[x]\psi') = (\psi,\boldsymbol{T}_{\mu \nu}[x]\psi')e_0^\mu e_0^\nu \quad \ \ (x = (t,p) \in \mathbb{R} \times \Sigma)
\end{align}
where $e_0^\mu = \alpha^{-1/2}\partial_t^\mu$ is the normalized timelike Killing vector field of the static spacetime. 

This brings us to giving a motivation of assumption (E) which we have highlighted as the property characteristic of an ``energy density''. One would generally expect 
that the energy density, integrated over a Cauchy surface of the static foliation we have at hand, yields the total energy. That means,
\begin{align} \label{eq:density-int}
 \int_\Sigma (\psi,\boldsymbol{\varrho}[t,p]\psi') \dif{\rm vol}_h(p) = (\psi,H\psi') 
\end{align}
for all $\psi,\psi' \in \mathcal{D}$ such that the integral exists.
Note that the right hand side is $t$-independent. Provided that 
$(\psi,\boldsymbol{\varrho}[t,p]\psi')$ and its derivatives vanish fast enough as 
``$p$ approaches $\infty$'', one can use the divergence-freeness of the stress-energy tensor together with Gauss' law to conclude that the integral on the left hand side of the previous equation is independent of $t$, so that the equation
is indeed consistent. 

Assuming now that \eqref{eq:density-int}
holds also for (suitable) vectors $\psi,\psi'$ contained in ${\sf A}_\infty(M)\Omega$, and that $G(t,p) = g_0(t)\underline{g}(p)$ and $O$ are chosen as 
stated in assumption (E), one obtains for $t \in [t_0 - \theta,t_0 + \theta]$ (cf.\ assumption (E)), and $A \in {\sf A}_\infty(O)$, 
\begin{align}
 (\psi,[H,A]\Omega) & = \int_\Sigma (\psi,[\boldsymbol{\varrho}[t,p],A]\Omega)\dif{\rm vol}_h(p) \label{eq:eins_}\\
 & = \int_\Sigma (\psi,[\boldsymbol{\varrho}[t,p],A]\Omega)\,\underline{g}(p)\dif{\rm vol}_h(p) \label{eq:zwei_} \\
 & = \int_{t_0 -\theta}^{t_0 + \theta}\int_\Sigma (\psi,[\boldsymbol{\varrho}[t,p],A]\Omega)\,\underline{g}(p)\dif{\rm vol}_h(p) \,g_0(t)\dif t \label{eq:drei_} \\
 & = (\psi,[\boldsymbol{\varrho}(G),A]\Omega)
\end{align}
where we could pass from \eqref{eq:eins_} to \eqref{eq:zwei_} since 
$\underline{g}$ is chosen such that it is equal to 1 on the support of 
$p \mapsto (\psi,[\boldsymbol{\varrho}[t,p],A]\Omega)$ for $t \in [t_0-\theta,t_0 + \theta]$. On the other hand, using that the stress-energy tensor has vanishing 
divergence, the integral in \eqref{eq:eins_} is $t$-independent 
(which is now 
a rigorous argument since the integrand is compactly supported in $p \in \Sigma$
for every $t$) which permits passage from \eqref{eq:zwei_} to \eqref{eq:drei_}
by the specifics of $g_0$. This serves to motivate how generally expected 
characteristic properties of an energy density in quantum field theory that are related to
obvious ``classical'' counterparts, such as \eqref{eq:stress_00} and \eqref{eq:density-int}, lead to the properties we have imposed on the energy-density quantum field in our assumption (E). While we expect that the line of our steps of motivation could be made rigorous --- certainly for linear quantized fields --- by a careful choice of state vectors $\psi$ and $\psi'$, our conditions set forth in (E) avoid the potential difficulties that may occur therein (like existence of 
the quadratic forms $(\psi,\boldsymbol{\varrho}[t,p]\psi')$, or of
the integral in \eqref{eq:density-int}) and  therefore are a more general way
of capturing the essential properties of a stress-energy observable in the present,
model-independent setting.

\medskip

For completeness and later use, we put on record a standard result, of which
similar variants can be found in \cite{BauWo} and in \cite{FreHer}. The notation is as follows. 
We choose $h\in C_0^\infty(\mathbb{R})$, $h \ge 0$ with ${\rm supp}(h) = [-1,1]$ and $\int_{-\infty}^\infty h(t)\dif t = 1$, and define the $\delta$-family $h_\kappa(t) = \kappa^{-1}h(t/\kappa)$ $(\kappa > 0,\ t \in \mathbb{R})$. 

\begin{Lemma} \label{Le:Hn-approx}
 Let $\xi \in \mathcal{H}$, with $\xi \ne 0$, be contained in 
 the domain of $(1 + H)^\nu$ for some $\nu \in \mathbb{N}_0$, and 
 let $O \subset M$ be a non-empty open subset such that the open interior 
 of $O^\perp$ is non-empty. Then there is a sequence $A_{(N)} \in {\sf A}_\infty(O)$
 $(N \in \mathbb{N})$ so that
 \begin{align}
  \| (1 + H)^\nu(\xi - A_{(N)}\Omega)\| \to 0 \quad (N \to \infty) \,.
 \end{align}
Moreover, the sequence  can be chosen such that $\| A_{(N)}\Omega\| = \|\xi\|$.
\end{Lemma}

\begin{proof}
We will show first that for any given $\eta > 0$ there is some $A = A_\eta$ so
that 
\begin{align}
  \| (1 + H)^\nu(\xi - A\Omega)\| < \eta\,.
 \end{align}
By the Reeh-Schlieder property (e), given
any non-empty open subset $O_1$ of $M$ with $\overline{O_1} \subset O$,
the vacuum vector $\Omega$ is cyclic and separating for ${\sf A}(O_1)$. 
With any such choice of $O_1$, and with the $\delta$-family $h_\kappa$ $(\kappa > 0)$
as defined previously, there is some $\kappa_0 > 0$ so that (cf.\ \eqref{eq:convol})
$u_{h_\kappa}(A_1) \in {\sf A}_\infty(O)$ whenever $A_1 \in {\sf A}(O_1)$
and $0 < \kappa < \kappa_0$. 

One can pick $0 < \kappa <\kappa_0$ with the property that 
\begin{align} \label{eq:Hn-est}
 \|(1+H)^\nu ( \xi -
 \hat{h}(\kappa H)\xi) \| < \frac{\eta}{2}
\end{align}
since 
\begin{align}
 \|(1+H)^\nu & ( \xi -
 \hat{h}(\kappa H)\xi) \| \\ & = \left\lVert\,\int_{-\infty}^\infty h_\kappa(t)( {\bf 1} - U_t) (1 + H)^\nu\xi \dif t \,\right\rVert \\ &  = \left( \int_{-\infty}^\infty h_\kappa(t)\dif t \right) \cdot \sup_{- \kappa \le t \le \kappa}\,\|(1 - U_t)(1 + H)^\nu \xi \| \to 0 \quad (\kappa \to 0)\,.
\end{align}
Moreover, making use of the Reeh-Schlieder property, one can choose $A_1\in {\sf A}(O_1)$
with 
\begin{align}
 \| \xi - A_1\Omega \| < \frac{\eta}{2 \| (1 + H)^\nu \hat{h}(\kappa H) \| + 1}
\end{align}
where the denominator on the right hand side displays the operator norm 
of $(1 + H)^\nu \hat{h}(\kappa H)$ for some (fixed) $0 < \kappa < \kappa_0$ that
has been picked to achieve the estimate \eqref{eq:Hn-est}. This is a bounded operator
since $\hat{h}$ is a Schwartz-type function. Hence, we obtain with 
$A = A_\eta = u_{h_\kappa}(A_1)$,
\begin{align}
 \| (1 + H)^\nu & (\xi - A\Omega)\| \\ & \le  \|(1+H)^\nu ( \xi -
 \hat{h}(\kappa H)\xi) \| + \| (1 + H)^\nu \hat{h}(\kappa H)(\xi - A_1\Omega)\|
 \nonumber
 \\ & < \frac{\eta}{2} + \frac{\eta}{2} = \eta\,. \nonumber
 \end{align}
This shows that there is a sequence $A_{[n]} \in {\sf A}_\infty(O)$ $(n \in \mathbb{N})$ such that 
\begin{align} \label{eq:simple}
 \| (1 + H)^\nu(\xi - A_{[n]}\Omega) \| \to 0 \quad (n \to \infty)\,.
\end{align}
On redefining $A_{(N)} = (\| \xi \|/\|A_{[N]}\Omega\|)A_{[N]}$ $(N \in \mathbb{N})$,
we have $\| A_{(N)} \Omega \| = \|\xi\|$. Moreover, \eqref{eq:simple}
implies $\|A_{[n]}\Omega\| \to \| \xi \|$ as $n \to \infty$. Thus, we 
obtain 
\begin{align}
 & \|(1  + H)^\nu (\xi - A_{(N)}\Omega)\| \\
 & \le \| (1 + H)^\nu(\xi - A_{[N]}\Omega)\| + \left| 1 - \frac{\|\xi\|}{\|A_{[N]}\Omega\|} \right| \cdot \|(1 + H)^\nu A_{[N]} \Omega \| \nonumber
\end{align}
where both terms on the right hand side converge to 0 as $N \to \infty$ on account of 
\eqref{eq:simple}. This proves the lemma.
\end{proof}

\section{Locally averaged quantum energy inequalities}
\label{sec:aqei}

In the following, we consider a quantum field theory with an energy density
on a static, globally hyperbolic spacetime, subject to the conditions described in 
Section \ref{sec:setting}. 

The first result we present is a simple variant of a quantum energy inequality for a local averaging of the energy density. It doesn't need all of the assumptions we have made; local preparability of states is not required. First, we introduce some notation. We assume that $G \in C_0^\infty(M,\mathbb{R})$ has been chosen as in assumption (E), and an open spacetime region $O$ so that the property \eqref{eq:dynamicgenerator} holds for all $A \in {\sf A}_\infty(O)$. Then we introduce, for any given $r \ge 1$, the subset $\mathcal{V}_r(O)$ given by 
\begin{align}
	\mathcal{V}_r(O) = \{ A\Omega : A \in {\sf A}_\infty(O)\,, \ \| A \|  \le r\,, \ \|A\Omega \| = 1\,\} \,.
\end{align}

\begin{Theorem} \label{Thm:First}
 Let $G$ and $O$ be as described above. Then for any $r\ge 1$ and every unit vector $\psi \in \mathcal{V}_r(O)$, the following estimate holds: 
\begin{align}
 (\psi,\boldsymbol{\varrho}(G)\psi) \ge  - r \|\boldsymbol{\varrho}(G) \Omega \|
\end{align}
\end{Theorem}

\begin{proof}
Since $\psi = A\Omega$ with $A \in {\sf A}_\infty(O)$, and the property 
\eqref{eq:dynamicgenerator} of the energy density, we have 
\begin{align} \label{eq:pos-split}
	(\psi,\boldsymbol{\varrho}(G)\psi) & = (A\Omega,\boldsymbol{\varrho}(G)A\Omega) \\
	& = (A\Omega,[\boldsymbol{\varrho}(G),A]\Omega) + (A\Omega,A \boldsymbol{\varrho}(G)\Omega)  \nonumber \\
	& = (A\Omega,H A\Omega) + (A\Omega,A \boldsymbol{\varrho}(G)\Omega)  \nonumber \\
	& \ge - \|A\Omega\| \cdot \|A\| \cdot \|\boldsymbol{\varrho}(G)\Omega\| \ge - r \|\boldsymbol{\varrho}(G)\Omega\| \nonumber
\end{align}
proving the claim. Note that we have used \eqref{eq:dynamicgenerator} and the fact that $H$ has non-negative spectrum.
\end{proof}

Despite the state-dependence of the lower bound on the energy density owing to the appearance of $r$ as a factor on the right-hand side in the previous theorem, it is worth remarking that the bound is non-trivial in the sense of not being an upper bound. In fact, we have the following statement:
\begin{Theorem} \label{Thm:2nd}
 For the quantized free scalar field on the given static, globally hyperbolic 
 spacetime, there is for any choice of $G$ and $O$ as in assumption (E) and every $\epsilon > 0$ a sequence of operators $B_m \in {\sf A}_\infty(O)$ $(m \in \mathbb{N})$ such that 
 $\| B_m \| \le 1 + \epsilon$, $\|B_m \Omega \| = 1$ and 
\begin{align}
 (B_m\Omega,\boldsymbol{\varrho}(G)B_m\Omega) \to \infty
 \quad (m \to \infty).
\end{align}
\end{Theorem}
\noindent The proof of this theorem will be given in Appendix \ref{appendix-A}.
Note that, by the Reeh-Schlieder property, $\bigcup_{r \ge 1} \mathcal{V}_r(O)$ is 
dense in the set unit vectors in $\mathcal{H}$; however, this is not true for $\mathcal{V}_r(O)$ for any 
fixed $r \ge 1$ under the assumptions we have made. Therefore, the lower bound 
on the energy density of Theorem \ref{Thm:First} is of limited use since, in order to 
apply to a set of unit vectors that is dense in the set of all unit vectors, formally
$r$ diverges to $\infty$. Such divergent behaviour is also supported by the results of 
\cite{ForHelRo}.\medskip

The next result, which is the main theorem of our work, is designed to overcome this shortcoming and gain some more control on a lower bound on the energy density for a dense set of vector states. To this end, we need to introduce some notation.

Again, we consider a test function 
$G \in C_0^\infty(M,\mathbb{R})$ and a non-empty open region $O$ so that the energy density
fulfills \eqref{eq:dynamicgenerator} from assumption (E). We also consider an arbitrary, causally convex open spacetime 
region $O^\times$ which contains ${\rm supp}(G)$, together
with another causally convex open spacetime region $O^\sharp$, so that $\overline{O^\times}
\subset O^\sharp$ and $O^\sharp$ admits a non-empty open causal complement. 
By the Reeh-Schlieder property (e), the vacuum vector $\Omega$ is cyclic and 
separating for ${\sf A}(O^\sharp)$. Therefore, there are the Tomita-Takesaki {\it modular conjugation} $J_\sharp$ and {\it modular operator} $\Delta_\sharp$ associated to the pair $({\sf A}(O^\sharp),\Omega)$ \cite{Tak,BraRob,Bor-Revol}, uniquely determined by the defining property 
\begin{align}
	J_\sharp \Delta^{1/2}_{\sharp} A\Omega = A^*\Omega \quad (A \in {\sf A}(O^\sharp))\,.
\end{align}
The corresponding {\it modular group} $\{\Delta^{is}_\sharp\}_{s \in \mathbb{R}}$ is commonly denoted as 
\begin{align}
	\Delta_\sharp^{is} = {\rm e}^{is K_{\sharp}} \, ,
\end{align}
where $K_\sharp = \log(\Delta_\sharp)$ is occasionally called the associated {\it modular Hamiltonian}.
\medskip

For $\lambda > 0$ we denote by $f_\lambda$ the scaled Gaussian,
\begin{align}
f_\lambda(s) = \frac{1}{\lambda}{\rm e}^{-(s/\lambda)^2/2} \quad (s \in \mathbb{R})\,,
\label{eq:flambda}
\end{align}
and $\hat{f}_\lambda$ denotes its Fourier transform, given by $\hat{f}_\lambda(k) = \int_{-\infty}^\infty {\rm e}^{isk} f_\lambda(s) \dif s = \sqrt{2\pi} {\rm e}^{-(\lambda k)^2/2}$ $(k \in \mathbb{R})$.

\begin{Theorem} \label{Thm:3rd}
 Let the test function $G$ and spacetime regions $O$, $O^\times$ and $O^\sharp$ be chosen
 with the properties described above. Then for any unit vector $\psi$ in the dense domain $\mathcal{D}$ and arbitrary $\epsilon > 0$, there is some $\lambda_0 > 0$ 
(depending on $\psi$ and $\epsilon$) such that 
\begin{align}
(\psi,\boldsymbol{\varrho}(G) \psi) \ge - \epsilon - \|\Delta_\sharp^{-1/2} \hat{f}_\lambda(K_\sharp)\boldsymbol{\varrho}(G)\Omega \| 
\end{align}
holds for all $0 < \lambda < \lambda_0$. 
\end{Theorem}

\begin{proof}
There is a causally convex open spacetime region $O^\flat$ such that 
$\overline{O^\times} \subset O^\flat$ and $\overline{O^\flat} \subset O^\sharp$ (see Figure \ref{fig}).
By the assumption (f) of local preparability of states there is some operator $Y \in {\sf A}(O^\flat)$ with $\| Y \| = 1$ 
such that 
\begin{align}
  (\psi,A\psi) = (Y\Omega,A Y \Omega) \quad (A \in {\sf A}(O^\times))\,.
\end{align}
As a consequence, on writing $G_t = G \circ \tau_t$, we obtain for sufficiently small $t'' > 0$, 
\begin{align}
 (\psi,\boldsymbol{\varrho}(G_t) \psi) & = 
 (Y\Omega, \boldsymbol{\varrho}(G_t)Y\Omega) \quad \text{and} \\
  (\psi,\boldsymbol{\varrho}(G_t)\boldsymbol{\varrho}(G_{t'}) \psi) & = 
 (Y\Omega, \boldsymbol{\varrho}(G_t)\boldsymbol{\varrho}(G_{t'})
 Y \Omega)
\end{align}
whenever $|t|,|t'| < t''$. This follows from the fact that, if $|t|, |t'| < t''$,
the operators $\boldsymbol{\varrho}(G_t)$ and $\boldsymbol{\varrho}(G_{t'})$ are affiliated with ${\sf A}(O^\times)$ and can be approximated by 
suitable bounded operators; e.g.\ writing 
$R_{t,\delta} = ({\bf 1} +  \delta |\boldsymbol{\varrho}(G_t)|)^{-1}$
for $\delta > 0$ yields 
\begin{align}
 & \lim_{\delta \to 0}\, (Y\Omega, R_{t,\delta}\boldsymbol{\varrho}(G_t)\boldsymbol{\varrho}(G_{t'})R_{t',\delta} 
 Y \Omega) \\ & = \lim_{\delta \to 0}\,(\psi,\boldsymbol{\varrho}(G_t)R_{t,\delta}R_{t',\delta}\boldsymbol{\varrho}(G_{t'})\psi) \nonumber \\
 & = (\psi,\boldsymbol{\varrho}(G_t)\boldsymbol{\varrho}(G_{t'}) \psi) \nonumber 
\end{align}
This also shows that 
$\boldsymbol{\varrho}(G_t)R_{t,\delta}Y\Omega$ as well as $R_{t,\delta}Y\Omega$ converge
as $\delta \to 0$ since 
\begin{align}
 \| \boldsymbol{\varrho}(G_t)(R_{t,\delta} - R_{t,\delta'})Y\Omega \|^2 = 
 \| \boldsymbol{\varrho}(G_t)(R_{t,\delta} - R_{t,\delta'})\psi \|^2 \to 0 \quad (\delta,\delta' \to 0)
\end{align}
and obviously $\|(R_{t,\delta} - {\bf 1})Y\Omega\| \to 0$ as $\delta \to 0$.
Thus, $Y\Omega$ lies in the domain of definition of (the selfadjoint extension of) $\boldsymbol{\varrho}(G_t)$ --- which is closed in the 
graph norm of $\boldsymbol{\varrho}(G_t)$ --- and similarly for $t'$ instead of $t$. Note that, for $|t| < t''$, 
\begin{align}
 \|\boldsymbol{\varrho}(G_t) Y\Omega\| = \| U_t^* \boldsymbol{\varrho}(G)U_t Y\Omega\| = \| \boldsymbol{\varrho}(G)U_t Y\Omega\|
\end{align}
by \eqref{eq:covariance-rho}, so that $U_tY\Omega$ is also in the domain of definition of $\boldsymbol{\varrho}(G)$ for $|t| < t''$.
Consequently, one obtains 
\begin{align}\label{eq:rho-cont-1}
 \| \boldsymbol{\varrho}(G)(U_t - {\bf 1})Y\Omega \| \to 0 \quad (t \to 0) 
\end{align}
since 
\begin{align}
 \| \boldsymbol{\varrho}(G)(U_t & - {\bf 1})Y\Omega \|  =
 \| U^*_t\boldsymbol{\varrho}(G)(U_t - {\bf 1})Y\Omega \| \\
 & = \|  \boldsymbol{\varrho}(G_t)Y\Omega - U_t^* \boldsymbol{\varrho}(G)Y\Omega \| \nonumber \\
 & \le \| ( \boldsymbol{\varrho}(G_t)-
 \boldsymbol{\varrho}(G))Y\Omega\| + 
 \|(U_t^* - {\bf 1})\boldsymbol{\varrho}(G)Y\Omega \| \nonumber \\
 & = \| ( \boldsymbol{\varrho}(G_t)-
 \boldsymbol{\varrho}(G))\psi\| + 
 \|(U_t^* - {\bf 1})\boldsymbol{\varrho}(G)Y\Omega \| \nonumber  \\
 & \to 0 \quad (t \to 0)\,. \nonumber
\end{align}
In what follows, we use a $\delta$-family $h_\kappa$ $(\kappa > 0)$ as introduced
prior to Lemma \ref{Le:Hn-approx}\,.
The property \eqref{eq:rho-cont-1} then allows to conclude
\begin{align} \label{eq:close-in}
 \| \boldsymbol{\varrho}(G)(u_{h_\kappa}(Y) - Y)\Omega\| \to 0 \quad (\kappa \to 0)\,.
\end{align}
To see this, we note that (assuming small enough $\kappa > 0$)
\begin{align}
 \| \boldsymbol{\varrho}(G)(u_{h_\kappa}(Y) & - Y)\Omega\|
 = \left\lVert \, \boldsymbol{ \varrho}(G) \int_{-\infty}^\infty h_\kappa(t) (U_t - {\bf 1})Y\Omega \dif t \,\right\rVert \\
 & \le \left(\int_{-\infty}^\infty h(t) \dif t \right) \sup_{|t| \le \kappa} \,\| \boldsymbol{\varrho}(G)(U_t  - {\bf 1})Y\Omega \| \to 0 \quad (\kappa \to 0) \nonumber
\end{align}
by the properties of the $h_\kappa$.
To simplify notation, we will from now on use the abbreviation
\begin{align}
 Y_{(\kappa)} = u_{h_\kappa}(Y) \quad (\kappa > 0) \,.
\end{align}
We observe that $\| Y_{(\kappa)}\| \le 1$. Furthermore, by \eqref{eq:close-in}
there is for any given $\epsilon >0$ some $\kappa_\epsilon > 0$ so that, for all $0 < \kappa \le \kappa_\epsilon$, one has 
\begin{align} \label{eq:1ststep}
 | \,(\psi,& \boldsymbol{\varrho}(G)\psi) - (Y_{(\kappa)}\Omega,\boldsymbol{\varrho}(G)Y_{(\kappa)}\Omega)\,| < \frac{\epsilon}{2} \\
 & \text{and} \quad Y_{(\kappa)} \in {\sf A}_\infty(O^\sharp)
\end{align}
We recall that, since $(1 + H)^{-\ell}\boldsymbol{\varrho}(G)(1 + H)^{-\ell}$
is bounded, it follows that $\boldsymbol{\varrho}(G)(1 + H)^{-2 \ell}$ is bounded. 
Since $Y_{(\kappa)}\Omega$ is in the $C^\infty$-domain of $H$, by Lemma 
\ref{Le:Hn-approx} that there is an $A \in {\sf A}_\infty(O)$, with 
$\|A\Omega\| = \|Y_{(\kappa)}\Omega\|$, fulfilling
\begin{align} \label{aprx-by-A}
 \| (1 + H)^{2 \ell}(Y_{(\kappa)} - A)\Omega\| < \frac{\epsilon}{8(1 + \|\boldsymbol{\varrho}(G)(1 + H)^{-2\ell}\|)} \,.
\end{align}
Consequently, one obtains 
\begin{align} \label{eq:2ndstep}
 | \, (Y_{(\kappa)}\Omega,\boldsymbol{\varrho}(G)Y_{(\kappa)}\Omega) - 
  (Y_{(\kappa)}\Omega,\boldsymbol{\varrho}(G)A\Omega)\,| < \frac{\epsilon}{8} \,.
\end{align}
In the next step we write, by virtue of \eqref{eq:dynamicgenerator},
\begin{align}
 (Y_{(\kappa)}\Omega,\boldsymbol{\varrho}(G)A\Omega) & = 
 (Y_{(\kappa)}\Omega,[\boldsymbol{\varrho}(G),A]\Omega) +
 (Y_{(\kappa)}\Omega,A\boldsymbol{\varrho}(G)\Omega) \nonumber \\ & =
 (Y_{(\kappa)}\Omega,HA\Omega) + (Y_{(\kappa)}\Omega,A\boldsymbol{\varrho}(G)\Omega) \,.
\end{align}
On the other hand, on account of \eqref{aprx-by-A},
\begin{align} \label{eq:3rdstep}
 |\,(Y_{(\kappa)}\Omega,HA\Omega) - (A\Omega,H A\Omega)\,| & = 
 |\,(H(Y_{(\kappa)} - A)\Omega,A\Omega)\,| < \frac{\epsilon}{8} \, ,
\end{align}
having used $\|A\Omega\| = \|Y_{(\kappa)}\Omega\| \le \|Y\Omega\| = 1$.
\medskip

\noindent Since the $\Delta_\sharp^{is} = {\rm e}^{is K_{\sharp}}$ $(s \in \mathbb{R})$ form a continuous unitary group, there is some $\lambda_0 > 0$
such that for $0 < \lambda < \lambda_0$ 
\begin{align} \label{eq:4thstep}
 \left|\, (Y_{(\kappa)}\Omega,A\boldsymbol{\varrho}(G)\Omega) - 
 \int_{-\infty}^\infty f_\lambda(s)\,(Y_{(\kappa)}\Omega, A \Delta_\sharp^{is} \boldsymbol{\varrho}(G)\Omega) \dif s \,\right|
  < \frac{\epsilon}{8}
\end{align}
with the scaled Gaussian $f_\lambda$ from \eqref{eq:flambda}. On the other hand, since both $Y_{(\kappa)}$ and $A$ are contained in ${\sf A}(O^{\sharp})$, and $J_\sharp$ and $\Delta^{1/2}_\sharp$ are the Tomita-Takesaki modular objects associated to the pair $({\sf A}(O^\sharp),\Omega)$, we now obtain
\begin{align}
 \int_{-\infty}^\infty f_\lambda(s)\,(Y_{(\kappa)}\Omega, A \Delta_\sharp^{is} \boldsymbol{\varrho}(G)\Omega) \dif s & = (A^*Y_{(\kappa)}\Omega, 
 \hat{f}_\lambda(K_\sharp)\boldsymbol{\varrho}(G)\Omega) \nonumber \\
 & = (J_\sharp \Delta^{1/2}_\sharp Y^*_{(\kappa)}A\Omega,
 \hat{f}_\lambda(K_\sharp)\boldsymbol{\varrho}(G)\Omega)  \nonumber \\
 & = (\Delta^{-1/2}_\sharp J_\sharp Y^*_{(\kappa)}A\Omega,
 \hat{f}_\lambda(K_\sharp)\boldsymbol{\varrho}(G)\Omega) \nonumber \\
& =  (J_\sharp Y_{(\kappa)}^* A\Omega,\Delta_\sharp^{-1/2}\hat{f}_\lambda(K_\sharp)\boldsymbol{\varrho}(G)\Omega)
\label{eq:rewrite-modular}
\end{align}
where the general property $J_\sharp \Delta^{1/2}_\sharp = \Delta^{-1/2}_\sharp J_\sharp$ for the Tomita-Takesaki modular objects has been used. 
Observing $\|J_\sharp Y_{(\kappa)}^*A\Omega\| \le 1$, we conclude that the 
modulus of the last term of the previous series of equations can be 
estimated by $\|\Delta_\sharp^{-1/2}\hat{f}_\lambda(K_\sharp)\boldsymbol{\varrho}(G)\Omega\|$.
\medskip

\noindent Finally, combining \eqref{eq:1ststep}, \eqref{eq:2ndstep}, \eqref{eq:3rdstep}
and \eqref{eq:4thstep},
we find
\begin{align}
 \mbox{\Large $|$} \,(\psi,\boldsymbol{\varrho}(G)\psi)  & - [\,(A\Omega,HA\Omega) + (J_\sharp Y^*_{(\kappa)}A\Omega, \Delta^{-1/2}_\sharp \hat{f}_\lambda(K_\sharp) \boldsymbol{\varrho}(G)\Omega)\,]\,\mbox{\Large $|$} \nonumber \\ & <
 \frac{\epsilon}{2} + \frac{\epsilon}{8} + \frac{\epsilon}{8} + \frac{\epsilon}{8} < \epsilon\,.
\end{align}
Since $(A\Omega,HA\Omega) \ge 0$, and by the previous estimate on the final expression $|(J_\sharp Y_{(\kappa)}^* A\Omega,\Delta_\sharp^{-1/2}\hat{f}_\lambda(K_\sharp)\boldsymbol{\varrho}(G)\Omega)|$ in \eqref{eq:rewrite-modular}, we arrive at 
\begin{align}
 (\psi,\boldsymbol{\varrho}(G)\psi) \ge - \epsilon - \|\Delta_\sharp^{-1/2}\hat{f}_\lambda(K_\sharp)\boldsymbol{\varrho}(G)\Omega\|
\end{align}
for all $0 < \lambda < \lambda_0$, as stated.
\end{proof}

\section{Concluding discussion}
\label{sec:discussion}

We have seen that in quantum field theory on any static, globally hyperbolic spacetime, relative lower bounds for certain types of spacetime averaged energy density can be established under very general, model-independent assumptions.\medskip

Some comments about the role of the parameters $\epsilon$ and $\lambda_0$ in Theorem \ref{Thm:3rd} are in order. As can be observed in the proof, any smaller choice
of $\epsilon$ will result in a smaller $\lambda_0$ in order for the statement 
of Theorem \ref{Thm:3rd} to be fulfilled. However, how small $\lambda_0$ must be made 
also depends on $\psi$, or rather, the choice of $A \in {\sf A}_\infty(O)$ with
\eqref{aprx-by-A}. As a matter of fact, similar to Theorem \ref{Thm:First}, its operator norm $\| A \|$ again is a controlling factor of 
a lower bound, since for any given $\epsilon' > 0$ it holds that 
\begin{align}
 \left| \, (Y_{(\kappa)}\Omega,A\boldsymbol{\varrho}(G)\Omega) - 
  \int_{-\infty}^\infty f_\lambda(s) (Y_{(\kappa)}\Omega A \Delta_\sharp^{is}\boldsymbol{\varrho}(G)\Omega) \dif s \,\right| < \epsilon' 
\end{align}
as soon as $\lambda$ is small enough so that 
\begin{align}
 \| ({\bf 1} - \hat{f}_\lambda(K_\sharp))\boldsymbol{\varrho}(G)\Omega \| < \frac{\epsilon'}{\| A \|} \,.
\end{align}
However, we think that controlling the lower bound on the spacetime averaged energy density $\boldsymbol{\varrho}(G)$ by the operator norm $\| A \|$ as in 
Theorem \ref{Thm:First} is very likely a rather crude estimate. In view of the 
fact that the operator norm of an operator $A \in {\sf A}_\infty(O)$ fulfilling
\eqref{aprx-by-A} is very hard to control, the result of Theorem \ref{Thm:3rd} 
offers a better controllable a priori lower bound on $\boldsymbol{\varrho}(G)$, in particular if the modular group $\{\Delta^{is}_\sharp\}_{s \in \mathbb{R}}$ acts geometrically like in the Bisognano-Wichmann theorem \cite{BiWi}. Nevertheless, the result of Appendix \ref{appendix-B} shows that one cannot expect a state-independent lower bound, i.e.\ $\lambda_0$ in the statement of Theorem \ref{Thm:3rd} is manifestly dependent on $\psi$ (or, more precisely, on the choices made for $Y_{(\kappa)}$ and $A$ in the proof).\medskip

Still, the universal form of the lower bound on the spacetime averaged density obtained in Theorem \ref{Thm:3rd} is remarkable. In the case that $\Delta^{is}_\sharp$ acts geometrically, the analogy of the result of Theorem
\ref{Thm:3rd} with a (relative) QWEI is even more palpable. Consider $\boldsymbol{\varrho}(G) = \boldsymbol{T}(G e_0 \otimes e_0)$ for stress-energy tensor $\boldsymbol{T}$ and normalized timelike Killing vector field $e_0$ (cf. \eqref{eq:stress_00}). Assume that there exists a timelike Killing flow $\Phi_s$ $(s \in \mathbb{R})$ on the region $O^\sharp$ of the underlying static spacetime so that 
\begin{align}
	\Delta_\sharp^{is} \boldsymbol{T}(G e_0 \otimes e_0) \Delta_\sharp^{-is} &= \boldsymbol{T}((\Phi_s)_*(G e_0 \otimes e_0)) \nonumber \\ &= \boldsymbol{T}((G\circ\Phi_{-s}) (\Phi_s)_*(e_0 \otimes e_0)) \, ,
\end{align}
where $(\Phi_s)_*$ denotes the pushforward of the diffeomorphism $\Phi_s$ on contravariant tensor fields of rank $2$ (see e.g. Appendix C of \cite{Wald-GR}). In such a case, one can choose, given any $\epsilon_1 > 0$, a $\lambda_0 > 0$ such that
\begin{align}
\label{eq:ineq-f-lambda}
 | (\psi,\boldsymbol{\varrho}(G)\psi) - (\psi, \boldsymbol{\varrho}_{f_\lambda} (G)\psi) | < \epsilon_1 \quad (0 < \lambda < \lambda_0)
\end{align}
for $f_\lambda$ from \eqref{eq:flambda}, where we have used the notation
\begin{align}
 \boldsymbol{\varrho}_{f_\lambda} (G) = \boldsymbol{T}\left( \int_{-\infty}^\infty f_\lambda(s)\, (G\circ\Phi_{-s}) (\Phi_s)_*(e_0 \otimes e_0) \dif s \right) \,,
\end{align}
with the integral to be interpreted in a suitable test function space, such as $\mathscr{S}(\mathbb{R}) \otimes C_0^\infty(\Sigma)$. Then the result of Theorem \ref{Thm:3rd} together with \eqref{eq:ineq-f-lambda} implies that, for given unit vector $\psi \in \mathcal{D}$ and $\epsilon > 0$, there is $\lambda_0 > 0$ so that 
\begin{align}
 (\psi,\boldsymbol{\varrho}_{f_\lambda} (G) \psi) \ge - \epsilon - 
 \| \Delta_\sharp^{-1/2} \boldsymbol{\varrho}_{f_\lambda} (G)\Omega\|
\end{align}
for all $0 < \lambda < \lambda_0$. This bears a striking similarity to a QWEI in which the expected energy density averaged along the trajectories of a congruence
of timelike curves is bounded below by a quantity which becomes negatively divergent if the averaging function (here: $f_\lambda$) becomes $\delta$-peaked (here: the limit $\lambda \to 0$). However, it actually resembles a relative QWEI bound since $\lambda_0$ depends on $\psi$, as discussed. Observe also that there is no
lower bound on the overall extension of the support of $G$. 

We also mention that the lower bounds established in Theorems \ref{Thm:First} and \ref{Thm:3rd} apply, in principle, to other local generators of the time translations in place of $\boldsymbol{\varrho}(G)$, like those obtained from the split property \cite{BuDoLo};\footnote{We thank Roberto Longo for pointing that out to us.} there may be some domain issues which would have to be clarified for asserting a rigorous result in such a case.\medskip

The lower bounds on the spacetime averaged energy established in the Theorems \ref{Thm:First} and \ref{Thm:3rd} use minimal, model-independent assumptions and have, as mentioned, a universal form. Thus their generality has the drawback of not being as specific as bounds that use model-dependent properties. In particular, they do not, for example, reveal the much stronger, state-independent QWEIs for the minimally coupled quantized Klein-Gor\-don field, or the Dirac field. The approach of splitting $(\psi,\boldsymbol{\varrho}(G)\psi) = (A\Omega,H A\Omega) + (A\Omega,A\boldsymbol{\varrho}(G)\Omega)$, as in \eqref{eq:pos-split} in the proof
of Theorem \ref{Thm:First}, into a positive part and a remainder is an obvious step in view of the assumptions made, but in the free field models it seems to take away too much of a positive contribution from the remainder term. 

One may hope that progress on a more detailed control of lower bounds on spacetime averaged energy densities in general quantum field theory may be made by combining operator product expansion techniques \cite{BosFew} and, e.g., conditions of modular nuclearity \cite{LeSa,BuDALo}.

\section*{Acknowledgements}

The authors would like to thank Daniela Cadamuro, Chris\ Fewster and Ko\ Sanders for discussions on topics of the present work. AGP thanks the IMPRS and MPI for Mathematics in the Sciences, Leipzig, for funding and support.

%%%%%%%%%%%%%%%%%%%%%%%%% APPENDIX %%%%%%%%%%%%%%%%%%%%%%%%%%%%

\appendix

\renewcommand{\thesubsection}{\Alph{subsection}}  % label appendix subsections alphabetically
\renewcommand{\theequation}{\thesubsection.\arabic{equation}}  % label equations in appendix after subsection
\numberwithin{Definition}{subsection}  % label theorem environments in appendix after subsection

\section*{Appendix}

\subsection{Proof of Theorem \ref{Thm:2nd}}
\label{appendix-A}

The quantization of the free scalar field on a static, globally hyperbolic spacetime
is very much established lore and we will be brief in our presentation. Standard references include \cite{Kay78,KayWald,Wald-QFTCST,Ful}. For the sake of notational simplicity we will restrict our attention to an ultrastatic spacetime, i.e.\ $M$ with metric $g$ given by \eqref{eq:metric} for $\alpha\equiv 1$, however the arguments for the general case of a static, globally hyperbolic spacetime are analogous up to a slightly more elaborate notation. The Klein-Gordon operator, viewed as mapping $C^\infty(M,\mathbb{R})$ 
into itself, is in the ultrastatic case given by ${\tt K} = \partial_t^2 - \Delta + \mu$
where $\Delta$ denotes the Laplace operator of the $d$-dimensional 
Riemannian manifold $(\Sigma,h)$ and $\mu \ge 0$ is a constant. We assume 
that $Q = - \Delta + \mu$ is invertible on $C_0^\infty(\Sigma,\mathbb{R})$ 
(``absence of zero modes'') which may require $\mu >0$ depending on $(\Sigma,h)$. The operator ${\tt K}$ admits unique advanced and retarded fundamental solutions ${\tt E}^{\rm av/rt}: C_0^\infty(M,\mathbb{R}) \to C^\infty(M,\mathbb{R})$, and their difference is 
${\tt E} = {\tt E}^{\rm av} - {\tt E}^{\rm rt}$. Then the factor space 
$L = C_0^\infty(M,\mathbb{R})/{\rm ker}({\tt E})$ is a symplectic space with 
symplectic form $\varsigma([F],[F']) = \int_M F(x) ({\tt E} F')(x)\dif{\rm vol}_g(x)$
for any $F,F' \in C_0^\infty(M,\mathbb{R})$, having denoted the canonical surjection
$C_0^\infty(M,\mathbb{R}) \to L$ by $F \mapsto [F]$ and the metric-induced volume form on 
$M$ by $\dif{\rm vol}_g$. There is also, for every $t \in \mathbb{R}$, the symplectic 
space ${\tt D}^{(t)} = C_0^\infty(\Sigma_t,\mathbb{R}) \oplus C_0^\infty(\Sigma_t,\mathbb{R})$ with symplectic form ${\tt d}^{(t)}( {\tt u} \oplus {\tt v},{\tt w} \oplus {\tt y}) = \int_{\Sigma_t} ({\tt uy} - {\tt vw})\dif{\rm vol}_h$ with the metric-induced volume form of $(\Sigma,h)$ as integration measure. For every $t \in \mathbb{R}$, there is a canonical symplectomorphism $P_t$ from $(L,\varsigma)$ to $({\tt D}^{(t)},{\tt d}^{(t)})$. 

The quantization proceeds by assigning to the symplectic space $(L,\varsigma)$ the 
$C^*$ Weyl algebra $\mathfrak{A}$ which is generated by a unit ${\bf 1}$ and 
a family of elements ${\bf w}([F])$ for $[F] \in L$, required to fulfill the relations
\begin{align}
 & {\bf w}([F])^*{\bf w}([F]) = {\bf 1}\,, \quad {\bf w}([F])^* = {\bf w}(-[F])\,, \\
 & {\bf w}([F]){\bf w}([F']) = {\rm e}^{-i\varsigma([F],[F'])/2} \, {\bf w}([F] + [F'])\,.
\end{align}
Local algebras are then obtained by defining $\mathfrak{A}(O)$ as the $C^*$-subalgebra 
of $\mathfrak{A}$ generated by the ${\bf w}([F])$ where ${\rm supp}(F) \subset O$, for
any open subset $O$ of $M$.

There is a quasifree ground state $\omega_0$ which is given by a complex scalar 
product $\Lambda_0$ over $L$. It is best described in terms of the symplectic 
space ${\tt D}^{(0)}$, i.e.\ we have chosen $t = 0$ for simplicity. 
The definition is 
\begin{align}
 \Lambda_0({\tt u} \oplus {\tt v},{\tt w}
\oplus {\tt y}) = \langle Q^{1/4}{\tt u} + iQ^{-1/4}{\tt v}, Q^{1/4}{\tt w} + iQ^{-1/4}{\tt y} \rangle
\end{align}
where 
\begin{align}
 \langle f,f'\rangle = \int_{\Sigma} \overline{f} f' \dif{\rm vol}_h \quad (f,f' \in L^2(\Sigma,\dif{\rm vol}_h)\,)
\end{align}
is the scalar product of the (complex) Hilbert space $L^2(\Sigma,\dif{\rm vol}_h)$.
One can check that ${\rm Im}\,\Lambda_0 = {\tt d}^{(0)}$, and therefore\footnote{At this point as well as in what follows, we identify $[F] \in L$
and ${\tt u} \oplus {\tt v} = P_0([F])$ mostly without explicitly writing the map
$P_0$ between $L$ and ${\tt D}^{(0)}$ in order to simplify the notation. We trust
that doing so won't lead to ambiguities.}
\begin{align}
 \omega_0({\bf w}({\tt u} \oplus {\tt v})) = {\rm e}^{-\Lambda_0({\tt u} \oplus {\tt v},{\tt u} \oplus {\tt v})/2}
\end{align}
defines a quasifree state $\omega_0$ on $\mathfrak{A}$. Its GNS representation $(\mathcal{H},\pi,\Omega)$ is a Fock representation \cite{KayWald}; thus, $\mathcal{H}$ equals the symmetric 
Fock space over the one-particle Hilbert space $\mathcal{H}_1 = L^2(\Sigma,\dif{\rm vol}_h)$, and $\Omega$ is the Fock vacuum vector. We denote the represented Weyl algebra
generators by 
\begin{align}
 W({\tt u} \oplus {\tt v}) = \pi({\bf w}({\tt u} \oplus {\tt v}))\,.
\end{align}
We write $\chi_{{\tt u} \oplus {\tt v}} = Q^{1/4}{\tt u} + iQ^{-1/4}{\tt v}$, and 
$a^+(\chi_{{\tt u} \oplus {\tt v}})$ for the creation operator of the vector 
$\chi_{{\tt u} \oplus {\tt v}} \in \mathcal{H}_1$. The notation for the annihilation
operator is similar, without superscript `$+$'. Using this notation, one finds
\begin{align}
 W({\tt u} \oplus {\tt v}) = {\rm e}^{i(a^+(\chi_{{\tt u} \oplus {\tt v}}) + 
 a(\chi_{{\tt u} \oplus {\tt v}}))}\,.
\end{align}
Furthermore, one can define the one-particle Hamilton operator $H_1 = Q^{1/2}$ and 
$H = d\Gamma(H_1)$, the second quantization of $H_1$ (in the notation of \cite{BraRob2}). The time evolution is on ${\tt D}^{(0)}$ given by $T_t P_0 [F] = P_0 [F \circ \tau_{-t}]$, and it holds that 
\begin{align}
 U_t W({\tt u} \oplus {\tt v})U_t^* = W(T_t({\tt u} \oplus {\tt v}))
\end{align}
with $U_t = {\rm e}^{itH}$. Clearly,  $U_t \Omega = \Omega$ holds, as well as 
$H \ge 0$. 

Given any open subset $S$ of $\Sigma_0$, one defines ${\sf A}(O)$ as the von Neumann
subalgebra of $\mathcal{B}(\mathcal{H})$ generated by the $W({\tt u} \oplus {\tt v})$
with ${\rm supp}({\tt u})$ and ${\rm supp}({\tt v})$ contained in $S$ for 
$O = D(S)$. This means ${\sf A}(O) = \pi(\mathfrak{A}(O))''$. 

The state $\omega_0$ is a quasifree Hadamard state, and also the 
coherent states $\omega_{({\tt u} \oplus {\tt v})}(\,\cdot\,) = (W({\tt u} \oplus {\tt v})\Omega,(\,\cdot\,)W({\tt u} \oplus {\tt v})\Omega)$
are Hadamard states. Here, we use $(\Psi,\Psi')$ to denote
the scalar product of vectors $\Psi,\Psi' \in \mathcal{H}$. For Hadamard states,
the expectation value of the quantized stress-energy tensor, and more specifically,
of the energy density, can be defined by a `point-splitting' procedure which agrees 
with the `normal ordering' prescription in the Fock space representation of the quasifree ground state we have at hand here. We will not discuss this at this point 
and instead refer to the references \cite{Wald-QFTCST,FV-Pass} where the matter is presented in detail. 

\begin{proof}[{Proof of Theorem \ref{Thm:2nd}}]
We can to a large amount rely on results of \cite{FV-Pass} (see, in particular, Appendix A of that work). It was shown in this reference that the energy density $\boldsymbol{\varrho}(G)$, defined according to the `point-splitting' procedure, exists as a quadratic form on Hadamard states, and fulfills 
\begin{align}
 &(W([F])\Omega,\boldsymbol{\varrho}(G)W([F])\Omega) \nonumber \\ & = (W([F])\Omega,HW([F])\Omega) + (W([F])\Omega,W([F])\boldsymbol{\varrho}(G)\Omega)
\end{align}
and
\begin{align}
 &(u_h(W([F]))\Omega,\boldsymbol{\varrho}(G)u_h(W([F]))\Omega) \nonumber \\ & = (u_h(W([F]))\Omega,H u_h(W([F]))\Omega) + (u_h(W([F]))\Omega,u_h(W([F]))\boldsymbol{\varrho}(G)\Omega)
\end{align}
provided that $W([F])$ and $u_h(W([F]))$ with 
$h \in C_0^\infty(\mathbb{R})$ and $F \in C_0^\infty(M,\mathbb{R})$
are contained in ${\sf A}(O)$, and 
with localization properties of $O$ relative to ${\rm supp}(G)$ as assumed 
for Theorem \ref{Thm:2nd}. Here, one uses the divergence-freeness of the 
expected stress-energy tensor, as mentioned in Section \ref{sec:assump}.
Furthermore, one can use (A.6) and (A.16) in \cite{FV-Pass} to conclude
that 
\begin{align}
 (W({\tt u} \oplus {\tt v})\Omega,H W({\tt u} \oplus {\tt v})\Omega) = 
 \frac{1}{2} \left( \langle {\tt u},Q{\tt u}\rangle + \langle{\tt v},{\tt v} \rangle \right)\,,
\end{align}
a well-known result. This means that we can choose, e.g., an open region 
$O_1$ with $\overline{O_1} \subset O$, and a sequence of unitaries 
$A_m = W(m{\tt u} \oplus m{\tt v})$ $(m \in \mathbb{N})$ in ${\sf A}(O_1)$ so that 
\begin{align}
 (A_m\Omega&,\boldsymbol{\varrho}(G)A_m\Omega) = (A_m\Omega,HA_m\Omega) + (A_m\Omega,A_m\boldsymbol{\varrho}(G)\Omega)
\nonumber  \\  & = 
 \frac{m^2}{2} \left( \langle {\tt u},Q{\tt u}\rangle + \langle{\tt v},{\tt v} \rangle \right) + (A_m^*A_m\Omega,\boldsymbol{\varrho}(G)\Omega) \to \infty \quad (m \to \infty)\,.
\end{align}
Using a variation of the arguments in Lemma \ref{Le:Hn-approx}, one can find a sequence of 
positive numbers $\kappa(m)$ converging to $0$ for $m \to \infty$ so that 
$\|(1 + H)(A_m - u_{h_{\kappa(m)}}(A_m))\Omega\| \to 0$ as $m \to \infty$.
Then the sequence of operators $u_{h_{\kappa(m)}}(A_m)$ ($m$ starting at high enough value) is 
contained in ${\sf A}_\infty(O)$ and is norm bounded by 1. 
Invoking again the arguments of Lemma \ref{Le:Hn-approx},
the sequence of 
operators $B_m = u_{h_{\kappa(m)}}(A_m)/\|u_{h_{\kappa(m)}}(A_m)\Omega\|$
is contained in ${\sf A}_\infty(O)$, and it fulfills $\|B_m\Omega\| = 1$
and $\| B_m \| \le 1 + \epsilon$ for sufficiently large $m$. It also fulfills
\begin{align}
 (B_m\Omega,\boldsymbol{\varrho}(G)B_m\Omega) = (B_m\Omega,HB_m\Omega) + (B_m^*B_m\Omega,\boldsymbol{\varrho}(G)\Omega) \to \infty \quad (m \to \infty)
\end{align}
in which the second term remains bounded in $m$. This proves the theorem. 
\end{proof}

\subsection{{$\protect\boldsymbol{\varrho}(G)\Omega$ is not in $\protect\mathrm{dom}(\Delta_\sharp^{-1/2})$}}
\label{appendix-B}

We give an argument illustrating that in general one cannot expect  $\boldsymbol{\varrho}(G)\Omega$ 
to be in the domain of $\Delta_\sharp^{-1/2}$ under the assumptions of 
Theorem \ref{Thm:3rd}.\footnote{The 3rd named author wishes to thank Daniele Guido for
having pointed out that fact, and outlined an argument to that effect, a long time ago.}

To this end, we consider a quantum field theory on $(1 + d)$-dimensional Minkowski spacetime, $M = \mathbb{R}^{1 + d}$, in a vacuum representation with Hilbert space 
$\mathcal{H}$, which carries a continuous unitary representation of the proper, orthochronous Poincar\'e group that acts covariantly on the family of local 
von Neumann algebras ${\sf A}(O)\subset\mathcal{B}(\mathcal{H})$ $(O \subset M)$, leaves the vacuum vector $\Omega\in\mathcal{H}$ invariant, and fulfills the relativistic spectrum condition. Let $W^R = \{(x^0,x^1,\ldots,x^d) \in M : 0 < |x^0| < x^1 \}$ denote the right wedge region, and denote by $\{U_s\}_{s \in \mathbb{R}}$, for $U_s = U(\Lambda_s^R)$, the one-parameter unitary subgroup of the representation of the proper, orthochronous Poincar\'e group implementing the Lorentz boosts 
\begin{align}
 & \Lambda_s^R (x^0,x^1,x^2,\ldots ,x^d) \\ & = (\cosh(s)x^0 - \sinh(s)x^1,-\sinh(s)x^0 + \cosh(s)x^1,x^2,\ldots,x^d) \nonumber
\end{align}
which leave $W^R$ invariant. 

The vacuum vector $\Omega$ is cyclic and separating for ${\sf A}(W^R)$. The Tomita-Takesaki modular group associated to the pair $({\sf A}(W^R),\Omega)$ will be denoted by $\{\Delta_R^{is}\}_{s \in \mathbb{R}}$ with $\Delta_R^{is}={\rm e}^{isK_R}$, and we assume that it acts geometrically as 
\begin{align} \label{eq:BiWi}
 \Delta_R^{is} = U_{2 \pi s} \quad (s \in \mathbb{R})\,. 
\end{align}
In other words, the result of the Bisognano-Wichmann theorem \cite{BiWi} is assumed. It holds whenever the local algebras are generated by bounded 
functions of a (scalar) quantum field. If that is the case, also the {\it timelike tube theorem} holds, which in our context means that if $O_1$ is any open subset of $M$ with 
$\overline{O_1} \subset W^R$, then the von Neumann algebra generated by 
all ${\sf A}(\Lambda_s^R(O_1))$, as $s$ ranges over $\mathbb{R}$, coincides with 
${\sf A}(W^R)$ \cite{Bor61}. Furthermore, if the $\Delta_R^{is}$ act geo\-metrically as in
\eqref{eq:BiWi}, and if the vacuum representation of the quantum field theory
is irreducible, meaning that ${\sf A}(M) = \mathcal{B}(\mathcal{H})$, then 
${\sf A}(W^R)$ is a factor, i.e.\ ${\sf A}(W^R) \cap {\sf A}(W^R)' = \mathbb{C}{\bf 1}$ \cite{Bor-Revol}.  Irreducibility of the vacuum representation is in fact a natural assumption, equivalent to
uniqueness of the vacuum vector up to phase. As mentioned, so are the other
assumptions entering the following proposition. 

In that proposition, $G \mapsto \boldsymbol{\varrho}(G)$ $(G \in C_0^\infty(\mathbb{R}^{1 + d}))$ is an operator-valued distribution fulfilling the previously given assumptions (A)--(D) for the 
case that $M = \mathbb{R}^{1 +d}$. (Assumption (E) is not required.) The 
Hamilton operator $H$ is the selfadjoint generator of the unitary representation
of the time shifts with respect to the inertial time coordinate $x^0$. (The unitary operators implementing the inertial, or ``static time direction'' shifts
have previously been denoted by $U_t = {\rm e}^{itH}$, which amounts to a slight abuse of notation given that the unitary operators implementing the Lorentz boosts are now denoted by $U_s = {\rm e}^{isK_R/2\pi}$. We trust that the reader won't be confused by this shift in notation.)

\begin{Proposition}
 Suppose that (i) the $\Delta^{is}_R$ act geometrically as in \eqref{eq:BiWi},
 (ii) the timelike tube theorem holds, and (iii) ${\sf A}(W^R)$ is a factor.
 \\[6pt]
If $\boldsymbol{\varrho}(G)\Omega$ is contained 
 in the domain of $\Delta_R^{-\delta}$ for some $\delta >0$,
 where $G \in C_0^\infty(\mathbb{R}^{1 +d},\mathbb{R})$ with ${\rm supp}(G) \subset W^R$, then 
 $\boldsymbol{\varrho}(G) = r{\bf 1}$ for some $r \in \mathbb{R}$.
\end{Proposition}

\begin{proof}
As $\boldsymbol{\varrho}(G)\Omega$ lies naturally in the domain of $\Delta_R^{1/2}$, and also 
lies in the domain of $\Delta_R^{-\delta}$ by assumption,
 it therefore lies in the domain of all $\Delta_R^\gamma$, $-\delta \le \gamma \le 1/2$ by an interpolation argument. Consequently, the $\mathcal{H}$-valued function 
$\zeta \mapsto \Delta_R^{-i\zeta}\boldsymbol{\varrho}(G)\Omega$ is analytic in the open strip 
$\Gamma_\delta = \{\zeta = s + i\gamma : s \in \mathbb{R}\,, \ -\delta < \gamma < 1/2\}$. 
Since ${\rm supp}(G) \subset W^R$, there is an open subset $O_1 \subset W^R$ so that,
for some $\epsilon > 0$, $\Lambda^R_s(O_1) \subset {\rm supp}(G)^\perp$ whenever $|s| < \epsilon$. Thus, for any finite choice of $A_1,\ldots,A_N$ in ${\sf A}_\infty(O_1)$,
$B \in {\sf A}_\infty(W^R)$ 
and $|s'_j| < \epsilon$ ($j = 1,\ldots,N \, , \, N\in\mathbb{N}$), one has 
\begin{align} \label{eq:tube}
(B\Omega,A_1(s'_1) \cdots A_N(s'_N) & \boldsymbol{\varrho}(G)\Omega)
 -(\boldsymbol{\varrho}(G)B\Omega,A_1(s'_1) \cdots A_N(s'_N)\Omega)
= 0\,, \nonumber \\
& \quad \text{where} \ A_j(s'_j) = U_{s'_j} A_j U_{-s'_j} \,.
\end{align}
Using the notation $s_1 = s'_1$, $s_{N+1} = -s'_N$, $s_j = s_j' - s_{j -1}'$ ($j = 2,\ldots, N$), the previous equation can be rewritten as 
\begin{align} \label{eq:bound.value}
  & (B\Omega,U_{s_1}A_1U_{s_2}A_2 \cdots U_{s_{N}}A_NU_{s_{N+1}}\boldsymbol{\varrho}(G)\Omega) \\
  & - (\boldsymbol{\varrho}(G)B\Omega,U_{s_1}A_1U_{s_2}A_2 \cdots U_{s_{N}}A_N\Omega) = 0 \, , \nonumber
\end{align}
which holds for all $A_j \in {\sf A}_\infty(O_1)$, $B \in {\sf A}_\infty(W^R)$, and 
all $s_j$ in a sufficiently small open interval around 0. 
Now we argue that this equation 
extends from $s_j$-values in a small open interval around 0 to all $s_j \in \mathbb{R}$.
 To see this, pick any $j$ between $1$ and $N+1$.
Then, for any $s_k \in \mathbb{R}$,
\begin{align}
 \psi = (U_{s_{1}} A_{1} \cdots U_{s_{j-1}} A_{j-1})^*B\Omega
\end{align}
can be written as $\psi = Q\Omega$ with a $Q \in {\sf A}(W^R)$. Hence $\psi$ lies in the domain of $\Delta_R^{1/2}$. We therefore have 
\begin{align} 
 & (B\Omega,U_{s_1}A_1U_{s_2}A_2 \cdots U_{s_j}A_j \cdots U_{s_{N}}A_NU_{s_{N+1}}\boldsymbol{\varrho}(G)\Omega) \nonumber \\
 & = (\psi,U_{s_j}A_j \cdots U_{s_{N}}A_NU_{s_{N+1}}\boldsymbol{\varrho}(G)\Omega) \nonumber \\
 & = (({\bf 1} + \Delta_R^{1/2})\psi,({\bf 1} + \Delta_R^{1/2})^{-1}U_{s_j}A_j \cdots U_{s_{N}}A_NU_{s_{N+1}}\boldsymbol{\varrho}(G)\Omega) \, ,
\end{align}
and we observe that 
\begin{align}
 s_j \mapsto ({\bf 1} + \Delta_R^{1/2})^{-1} U_{s_j} 
\end{align}
is the strong boundary value, for negative imaginary part tending to 0, of the 
operator-valued function 
\begin{align}
 \zeta \mapsto ({\bf 1} + {\rm e}^{K_R/2})^{-1} {\rm e}^{i\zeta K_R / 2\pi} \, ,
\end{align}
which is strongly analytic in the open strip $\{-1/2<{\rm Im}(\zeta)<0\} \subset 
\mathbb{C}$.
The same conclusion applies when replacing $\psi$ by
\begin{align}
\tilde{\psi} = \tilde{Q}\Omega = (U_{s_{1}} A_{1} \cdots U_{s_{j-1}} A_{j-1})^*
\boldsymbol{\varrho}(G)B\Omega
\end{align}
where $\tilde{Q}$ is a closable operator defined on ${\sf A}_\infty(\mathbb{R}^{1+d})\Omega$ that is affiliated with ${\sf A}(W^R)$. 
Moreover, $\zeta \mapsto \Delta_R^{-i\zeta}\boldsymbol{\varrho}(G)\Omega$ is analytic in the open strip
$\Gamma_\delta$ around the real axis as argued before.
Therefore, we can conclude iteratively, starting with $j = 1$ and continuing up to $j = N+1$,
that equation \eqref{eq:bound.value} extends from $s_j$ taken 
from an open interval around 0 to all $s_j \in \mathbb{R}$ ($j = 1,\ldots, N+1$). This implies that \eqref{eq:tube} extends to all $s'_j \in \mathbb{R}$ ($j = 1,\ldots, N$) so that,
employing the timelike tube theorem, we obtain
\begin{align} \label{eq:commut}
[A,\boldsymbol{\varrho}(G)]\Omega = 0
\end{align}
for all $A \in {\sf A}_\infty(W^R)$. 
Using the arguments of \cite{FreHer}, one can check that if $\boldsymbol{\varrho}(G)$
is an operator of an $H$-bounded quantum field, then 
also every monomial $\boldsymbol{\varrho}(G)^q$, $q \in \mathbb{N}$,
is $H$-bounded, i.e.\ there is for every 
$q \in \mathbb{N}$ some $\ell(q) \in \mathbb{N}$ with the property that 
$\boldsymbol{\varrho}(G)^q(1 + H)^{-\ell(q)}$ is bounded. Therefore, \eqref{eq:commut}
implies that 
\begin{align}
 [A,\boldsymbol{\varrho}(G)^q]\Omega = 0 \quad \ \ (q \in \mathbb{N})
\end{align}
for all $A \in {\sf A}_\infty(W^R)$. This is shown by induction on $q$, noting that 
$[A,\boldsymbol{\varrho}(G)^q]\Omega$ = 0 implies 
$0 = B[A,\boldsymbol{\varrho}(G)^q]\Omega =[A,\boldsymbol{\varrho}(G)^q]B\Omega$ for 
all $B \in {\sf A}_\infty(W^L)$, where $W^L =\{(x^0,\ldots,x^d): x^1 < 0\,,\ |x^0| < |x^1|\} = {\rm int}(W^R)^\perp$ is the left wedge region, and using that $\boldsymbol{\varrho}(G)^q\Omega$ is in the $C^\infty$-domain of $H$ together with Lemma \ref{Le:Hn-approx}.
Then the previous equation implies for every vector 
$\chi \in {\rm dom}(\boldsymbol{\varrho}(G)^q)$,
\begin{align}
 (\chi,A\boldsymbol{\varrho}(G)^q\Omega) - (\boldsymbol{\varrho}(G)^q\chi,A \Omega) = 0
 \quad \ \ (q \in \mathbb{N}\,, \ A \in {\sf A}_\infty(W^R)).
\end{align}
For every finite interval spectral projector $E$ of $\boldsymbol{\varrho}(G)$ and $A \in {\sf A}(W^R)$ we can find a sequence in ${\sf A}_\infty(W^R)$ converging strongly to $EAE$ and thus we obtain from the last equation 
\begin{align} \label{eq:mon-proj}
 EAE\boldsymbol{\varrho}(G)^q - \boldsymbol{\varrho}(G)^q EAE = 0 \quad \ \ (q \in \mathbb{N}\,, \ A \in {\sf A}(W^R))\,,
\end{align} 
where we have used that the vector $\Omega$ is separating for ${\sf A}(W^R)$.
Assume that $E = E_n$ is the spectral projector corresponding to the 
spectral interval $[-n,n]$ of $\boldsymbol{\varrho}(G)$ for $n \in \mathbb{N}$ and let, for
given $a > 0$,
$P_{a,\nu}$ $(\nu \in \mathbb{N})$ be a sequence of polynomials approaching 
$\lambda \mapsto ({\bf 1} + a|\lambda|^2)^{-1}\lambda$ uniformly for $\lambda \in [-n,n]$.
Then we conclude from \eqref{eq:mon-proj} that for every $A \in {\sf A}(W^R)$,
\begin{align}
0 = \lim_{\nu \to \infty}\, [E_nAE_n,P_{a,\nu}(\boldsymbol{\varrho}(G))] =  \left[E_nAE_n,\frac{\boldsymbol{\varrho}(G)}{{\bf 1} + a|\boldsymbol{\varrho}(G)|^2}\right]\,.
\end{align}
This holds for arbitrary $n \in \mathbb{N}$ and $a >0$. Taking the limit $n \to \infty$, we find
that for any $a > 0$, 
\begin{align}
 \frac{\boldsymbol{\varrho}(G)}{{\bf 1} + a|\boldsymbol{\varrho}(G)|^2} \in {\sf A}(W^R) \cap {\sf A}(W^R)' = \mathbb{C}{\bf 1}\,.
\end{align}
Since 
\begin{align}
 \lim_{a \to 0} \left(\frac{\boldsymbol{\varrho}(G)}{{\bf 1} + a|\boldsymbol{\varrho}(G)|^2}\right) \psi = \boldsymbol{\varrho}(G)\psi 
\end{align}
holds for all $\psi \in {\rm dom}(\boldsymbol{\varrho}(G))$, the statement of the Proposition is hence proved.
\end{proof}

\end{document}